\documentclass[draftclsnofoot,onecolumn,12pt]{IEEEtran}
\usepackage[blocks]{authblk}
\usepackage{url}
\usepackage{flushend}
\usepackage{graphicx}
\usepackage{amsmath}
\usepackage{setspace}
\usepackage{verbatim}
\usepackage{amsthm}
\usepackage{bbm}
\usepackage{amssymb}
\usepackage{psfrag}
\usepackage{booktabs}
\usepackage{multirow}

\newtheorem{thm}{Theorem}
\newtheorem{cor}[thm]{Corollary}
\newtheorem{lem}[thm]{Lemma}
\newtheorem{defn}[thm]{Definition}
\newtheorem{ex}[thm]{Example}

\newcommand{\bX}{\mathbf{X}}
\newcommand{\bY}{\mathbf{Y}}
\newcommand{\bZ}{\mathbf{Z}}
\newcommand{\bW}{\mathbf{W}}

\newcommand{\bM}{\mathbf{M}}
\newcommand{\bx}{\mathbf{x}}
\newcommand{\by}{\mathbf{y}}
\newcommand{\bz}{\mathbf{z}}

\newcommand{\cA}{\mathcal{A}}
\newcommand{\cC}{\mathcal{C}}
\newcommand{\cX}{\mathcal{X}}
\newcommand{\cY}{\mathcal{Y}}
\newcommand{\cZ}{\mathcal{Z}}
\newcommand{\cF}{\mathcal{F}}
\newcommand{\cP}{\mathcal{P}}
\newcommand{\cW}{\mathcal{W}}
\newcommand{\cFen}{\mathcal{F}_{\epsilon}^{n}}

\newcommand{\ph}{\hat{p}}
\newcommand{\Gh}{\hat{G}}
\newcommand{\Eh}{\hat{E}}
\newcommand{\Vh}{\hat{V}}

\newcommand{\fh}{\hat{f}}

\newcommand{\limtoi}{\lim_{n\to\infty}}

\newcommand{\Teen}{T_{\frac{\epsilon}{2}}^n}
\newcommand{\Ten}{T_{\epsilon}^n}

\newcommand{\cS}{\mathcal{S}}

\newcommand{\cxo}{c_{G_{X_1}}}
\newcommand{\cxomin}{c_{G_{X_1}}^{min}}

\newcommand{\cxt}{c_{G_{X_2}}}
\newcommand{\cxk}{c_{G_{X_k}}}

\newcommand{\Cxo}{C_{G_{X_1}}}

\newcommand{\cxon}{c_{G_{\bX_1}^{n}}}

\newcommand{\cxtn}{c_{G_{\bX_2}^{n}}}
\newcommand{\cxkn}{c_{G_{\bX_k}^{n}}}
\newcommand{\cxsn}{c_{G_{\bX_S}^{n}}}
\newcommand{\cxscn}{c_{G_{\bX_{S^c}}^{n}}}

\newcommand{\Cxon}{C_{G_{\bX_1}^{n}}}
\newcommand{\Cxtn}{C_{G_{\bX_2}^{n}}}

\newcommand{\cjmin}{c_{G_{\bX_1}^{n},G_{\bX_2}^{n}}^{min}}
\newcommand{\cjsmin}{c_{G_{\bX_1}^{n},G_{\bX_2}^{n}}'}

\newcommand{\gxo}{G_{X_1}}
\newcommand{\gxt}{G_{X_2}}
\newcommand{\gxk}{G_{X_k}}

\newcommand{\sgjt}{G_{X_1},G_{X_2}}
\newcommand{\sgjk}{G_{X_1},...,G_{X_k}}
\newcommand{\sgji}{G_{X_1},...,G_{X_i}}
\newcommand{\sgs}{[G_{X_i}]_{i\in S}}
\newcommand{\sgxi}{[G_{\Xi_{iz}}]_{z\in s_i}}

\newcommand{\exo}{E_{X_1}}

\newcommand{\exon}{E_{X_1}^{n}}
\newcommand{\vxon}{V_{X_1}^{n}}

\newcommand{\gxon}{G_{\bX_1}^{n}}
\newcommand{\gxtn}{G_{\bX_2}^{n}}
\newcommand{\gxkn}{G_{\bX_k}^{n}}

\newcommand{\pij}{\Xi_{ij}}
\newcommand{\pis}{\Xi_{is_{i}}}
\newcommand{\pisc}{\Xi_{is_{i}^{c}}}

\newcommand{\gm}{G_{X,f_{1},f_{2}}}
\newcommand{\gf}{G_{X,f_{1}}}
\newcommand{\gff}{G_{X,f_{2}}}
\newcommand{\gmm}{G_{X,f_{1},...,f_{m}}}

\newcommand{\gxofh}{G_{X_1,\fh}}
\newcommand{\gxtfh}{G_{X_2,\fh}}

\newcommand{\cf}{c_{G_{X,f_{1}}^{n}}}

\newcommand{\cm}{c_{G_{X,f_{1},f_{2}}^{n}}}

\newcommand{\eps}{\epsilon}

\begin{document}

\title{On Network Functional Compression}

\author{Soheil~Feizi,~\IEEEmembership{Student~Member,~IEEE}, Muriel~M\'edard,~\IEEEmembership{Fellow,~IEEE}%
	\thanks{Soheil Feizi and Muriel M\'edard are with the Research Laboratory of Electronics (RLE) at Massachusetts Institute of Technology (MIT), Cambridge, MA.}%
	\thanks{Portions of this article were presented at the \textit{2009 Annual Allerton Conference on Communication, Control, and Computing} at University of Illinois, Urbana-Champaign, the \textit{2010 IEEE International Symposium on Information Theory (ISIT 2010)} in Austin, TX, and the \textit{IEEE Globecom 2009 Communication Theory Symposium} in Honolulu, HI.}%
	\thanks{The authors acknowledge the support of the Jacobs Presidential Fellowship and AFOSR award number FA9550-09-1-0196.}
	}        

\maketitle

\begin{abstract}
In this paper, we consider different aspects of the problem of compressing for function computation across a network, which we call \textit{network functional compression}. In network functional compression, computation of a function (or, some functions) of sources located at certain nodes in a network is desired at receiver(s), that are other nodes in the network. The rate region of this problem has been considered in the literature under certain restrictive assumptions, particularly in terms of the network topology, the functions and the characteristics of the sources. In this paper, we present results that significantly relax these assumptions. Firstly, we consider this problem for an arbitrary tree network and asymptotically lossless computation and derive rate lower bounds. We show that, for depth one trees with correlated sources, or for general trees with independent sources, a modularized coding scheme based on graph colorings and Slepian-Wolf compression performs arbitrarily closely to rate lower bounds. For a general tree network with independent sources, optimal computation to be performed at intermediate nodes is derived. We show that, for a family of functions and random variables called chain rule proper sets, computation at intermediate nodes is not necessary. We introduce a new condition on colorings of source random variables' characteristic graphs called the \textit{coloring connectivity condition (C.C.C.)}. We show that, this condition is necessary and sufficient for any achievable coding scheme based on colorings, thus relaxing the previous sufficient zig-zag condition of Doshi \textit{et al.} We also show that, unlike entropy, graph entropy does not satisfy the chain rule. 

Secondly, we consider a multi-functional version of this problem with side information, where the receiver wants to compute several functions, with different side information random variables. We derive a rate region and propose a coding scheme based on graph colorings for this problem. Thirdly, we consider the functional compression problem with feedback. We show that, in this problem, unlike Slepian-Wolf compression, by having feedback, one may outperform rate bounds of the case without feedback. These results extend those of Bakshi \textit{et al.} Fourthly, we investigate the problem of distributed functional compression with distortion, where computation of a function within a distortion level is desired at the receiver. We compute a rate-distortion region for this problem. Then, we propose a simple suboptimal coding scheme with a non-trivial performance guarantee.

Our coding schemes are based on finding the minimum entropy coloring of the characteristic graph of the function we seek to compute. In general, it is shown by Cardinal \textit{et al.} that finding this coloring is an NP-hard problem. However, we show that, depending on the characteristic graph's structure, there are some interesting cases where finding the minimum entropy coloring is not NP-hard, but tractable and practical. In one of these cases, we show that, by having a non-zero joint probability condition on random variables' distributions, for any desired function, finding the minimum entropy coloring can be solved in polynomial time. In another case, we show that, if the desired function is a quantization function with a certain structure, this problem is also tractable. 

\end{abstract}


\section{Introduction}\label{chap:intro}

In this paper, we consider different aspects of the functional compression problem over networks. In the functional compression problem, we would like to compress source random variables for the purpose of computing a deterministic function (or some deterministic functions) at the receiver(s) when these sources and receivers are nodes in a network. Traditional data compression schemes are special cases of functional compression, where their desired function is the identity function. However, if the receiver is interested in computing a function (or some functions) of sources, further compressing is possible. In the rest of this section, we review some prior relevant research and illustrate some research challenges of this problem through some motivating examples which will be discussed in the following sections.

\subsection{Prior Work in Functional Compression}

We categorize prior work into the study of lossless functional compression and that of functional compression with distortion.  

\subsubsection{Lossless Functional Compression}

By lossless computation, we mean asymptotically lossless computation of a function: the error probability goes to zero as the block length goes to the infinity. 

  \begin{figure}[t]
	\centering
    \includegraphics[width=10.5cm,height=8cm]{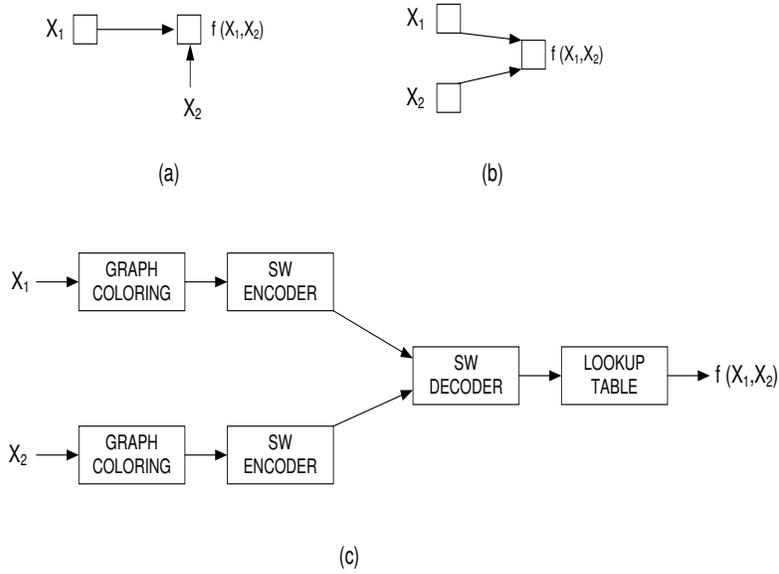}
    \caption{a) Functional compression with side information b) A distributed functional compression problem with two transmitters and a receiver c) An achievable encoding/decoding scheme for the functional compression.}
    \label{fig:3part}
  \end{figure}

First, consider the network topology depicted in Figure \ref{fig:3part}-a which has two sources and a receiver. One of the sources is available at the receiver as the side information. Shannon was the first one who considered this problem in \cite{s56} for a special case when $f(X_1,X_2)=(X_1,X_2)$ (the identity function). For a general function, Orlitsky and Roche provided a single-letter characterization in \cite{or01}. In \cite{doshi-it}, Doshi et al. proposed an optimal coding scheme for this problem.

Now, consider the network topology depicted in Figure \ref{fig:3part}-b which has two sources and a receiver. This problem is a distributed compression problem. For the case that the desired function at the receiver is the identity function (i.e., $f(X_1,X_2)=(X_1,X_2)$), Slepian and Wolf provided a characterization of the rate region and an optimal achievable coding scheme in \cite{sw73}. Some other practical but suboptimal coding schemes have been proposed by Pradhan and Ramchandran in \cite{pradhan}. Also, a rate-splitting technique for this problem is developed by Coleman \textit{et al.} in \cite{coleman}. Special cases when $f(X_1,X_2)=X_1$ and $f(X_1,X_2)=(X_1+X_2)\mod 2$ have been investigated by Ahlswede and K\"orner in \cite{korner-ahl}, and K\"orner and Marton in \cite{korner-marton}, respectively. Under some special conditions on source distributions, Doshi \textit{et al.} in \cite{doshi-it} investigated this problem for a general function and proposed some achievable coding schemes.   

  \begin{figure}[t]
	\centering
    \includegraphics[width=5.5cm,height=5.5cm]{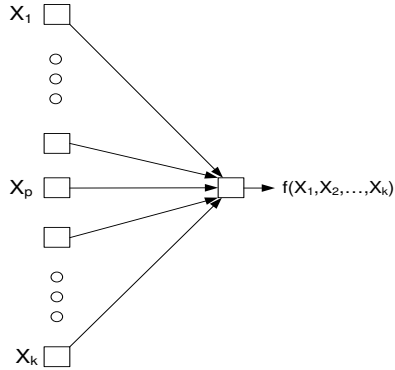}
    \caption{ A general one-stage tree network with a desired function at the receiver.}
    \label{fig:one-stage}
  \end{figure}

Sections \ref{chap:tree}, \ref{chap:multi} and \ref{chap:feedback} of this paper consider different aspects of this problem (asymptotically lossless functional compression). In particular, we are going to answer to the following questions in these sections:

\begin{itemize}
\item For a depth one tree network with one desired function at the receiver (as shown in Figure \ref{fig:one-stage}), what is a necessary and sufficient condition for any coding scheme to guarantee that the network is solvable (i.e., the receiver is able to compute its desired function)?  
\item What is a rate region of the functional compression problem for a depth one tree network (a rate region is a set of rates for different links of the network under which the network is solvable)? How can a modularized coloring-based coding scheme perform arbitrarily closely to rate bounds?
\item For a general tree network with one desired function at the receiver (as shown in Figure \ref{fig:general-tree}), when do intermediate nodes need to perform computation and what is an optimal computation to be performed? What is a rate-region for this network?  
\item How do results extend to the case of having several desired functions with the side information at the receiver? 
\item What happens if we have feedback in our system?
\end{itemize}

  \begin{figure}[t]
	\centering
    \includegraphics[width=10.5cm,height=7cm]{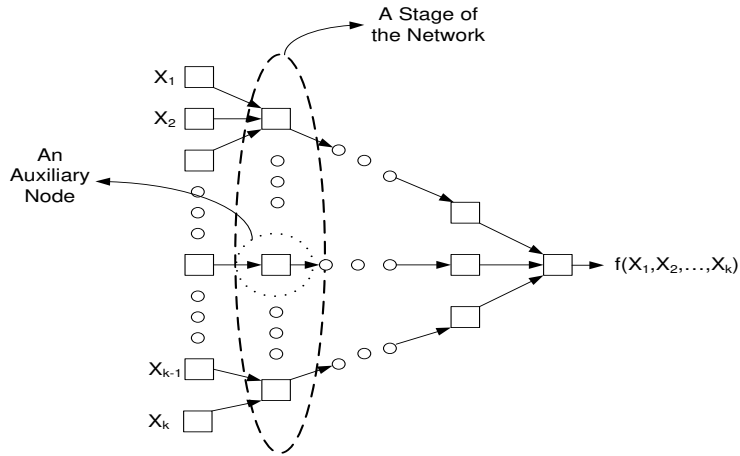}
    \caption{ An arbitrary tree network topology.}
    \label{fig:general-tree}
  \end{figure}

\subsubsection{Functional Compression with Distortion}

In this section, we review prior results in functional compression for the case of being allowed to compute the desired function at the receiver within a distortion level. 

First, consider the network topology depicted in Figure \ref{fig:3part}-a, with the side information at the receiver. Wyner and Ziv \cite{wz76} considered this problem for computing the identity function at the receiver with distortion $D$. Yamamoto solved this problem for a general function $f(X_1,X_2)$ in \cite{yam}. Doshi et al. gave another characterization of the rate distortion function given by Yamamoto in \cite{doshi-it}. Feng et al. \cite{fes04} considered the side information problem for a general function at the receiver in the case the encoder and decoder have some noisy information.  

For the network topology depicted in \ref{fig:3part}-b and for a general function, the rate-distortion region has been unknown, but some bounds have been given by Berger and Yeung \cite{by89}, Barros and Servetto \cite{bs03}, and Wagner et al. \cite{wtv06}, where considered a specific quadratic distortion function. 

In Section \ref{chap:distortion} of this paper, we answer to the following questions:

\begin{itemize}
\item What is a multi-letter characterization of a rate-distortion function for a distributed network depicted in Figure \ref{fig:3part}-b? 
\item For this problem, is there a simple suboptimal coding scheme based on graph colorings with a non-trivial performance guarantee?
\end{itemize}

\textit{Remarks:} 
\begin{itemize}
\item Our coding schemes in Sections [\ref{chap:intro}-\ref{chap:distortion}] are based on finding a minimum entropy coloring of a characteristic graph of the function we seek to compute. In general, reference \cite{np1} showed that, finding this coloring is an NP-hard problem. However, in Section \ref{chap:min-coloring}, we consider whether there are some functions and/or source structures which lead to easy and practical coding schemes.

\item Note that, our work is different in techniques and the problem setup from multi-round function computation (e.g., \cite{ish1} and \cite{ish2}). Also, some references consider the functional computation problem for specific functions. For example, \cite{kumar} investigated computation of symmetric Boolean functions in tree networks and \cite {sagar} and \cite{ramm} studied the sum-network with three sources and three terminals. Note that, in our problem setup, the desired function at the receiver is an arbitrary function. Also, we are interested in asymptotically lossless or lossy computation of this function.  
\end{itemize}

In the rest of this section, we explain some research challenges of the functional compression problem by some examples. In the next sections, we explain these issues with more detail.

\subsection{Problem Outline}\label{sec:challenge}

In this section, we address some problem outlines of functional compression. We use different simple examples to illustrate these issues which will be explained later in this paper.

Let us proceed by an example:

\begin{ex}\label{ex:mod2}
Consider the network shown in Figure \ref{fig:3part}-b, which has two source nodes and a receiver. Suppose source nodes have two independent source random variables (RVs) $X_1$ and $X_2$ such that $X_1$ takes values from the set $\cX_1=\{x_1^1,x_1^2,x_1^3,x_1^4\}=\{0,1,2,3\}$, and $X_2$ takes values from the set $\cX_2=\{x_2^1,x_2^2\}=\{0,1\}$, both with equal probability. Values of $x_i^j$ for different $i$ and $j$ are shown in Figure \ref{fig:char-graph-sq}. Suppose the receiver desires to compute a function $f(X_1,X_2)=(X_1+X_2)\mod 2$. 

If $X_1=0$ or $X_1=2$, for all possible values of $X_2$, we have $f(X_1,X_2)=X_2$. Hence, we do not need to \textit{distinguish} between $X_1=0$ and $X_1=2$. A similar argument holds for $X_1=1$ and $X_1=3$. However, cases $X_1=0$ and $X_1=1$ should be distinguished, because for $X_2=0$, the function value is different when $X_1=0$ than the one when $X_1=1$ (i.e., $f(0,0)=0\neq f(1,0)=1$). 
\end{ex}

  \begin{figure}[t]
	\centering
    \includegraphics[width=8.5cm,height=5cm]{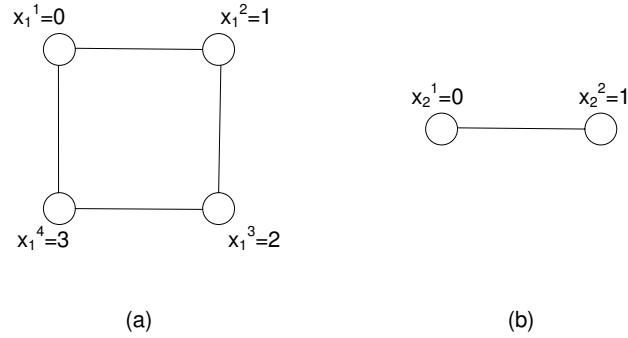}
    \caption{ Characteristic graphs described in Example \ref{ex:mod2}: a) $G_{X_1}$, b) $G_{X_2}$.}
    \label{fig:char-graph-sq}
  \end{figure}

We notice that for each source random variable, depending on the function at the receiver and values of the other source random variable, we should \textit{distinguish} some possible pair values. In other words, values of source random variables which \textit{potentially} can cause confusion at the receiver should be assigned to different \textit{codes}. To determine which pair values of a random variable should be assigned to different codes, we make a graph for each random variable, called the \textit{characteristic graph} or the \textit{confusion graph} of that random variable (\cite{s56}, \cite{k73}). Vertices of this graph are different possible values of that random variable. We connect two vertices if they should be distinguished. For the problem described in Example \ref{ex:mod2}, the characteristic graph of $X_1$ (called $G_{X_1}$) is depicted in Figure \ref{fig:char-graph-sq}-a. Note that we have not connected vertices which lead to the same function value for all values of $X_2$. The characteristic graph of $X_2$ ($G_{X_2}$) is shown in Figure \ref{fig:char-graph-sq}-b.

Now, we seek to assign different codes to connected vertices, which corresponds a \textit{graph coloring} where we assign different colors (codes) to connected vertices. Vertices that are not connected to each other can be assigned to the same or different colors (codes). Figure \ref{fig:look-up-ex}-(a,b) shows valid colorings for $G_{X_1}$ and $\gxt$.

Now, we propose a possible coding scheme for this example. First, we choose valid colorings for $\gxo$ and $\gxt$. Instead of sending source random variables, we send these coloring random variables. At the receiver side, we use a look-up table to compute the desired function value by using the received colorings. Figure \ref{fig:look-up-ex} demonstrates this coding scheme.

  \begin{figure}[t]
	\centering
    \includegraphics[width=11cm,height=11cm]{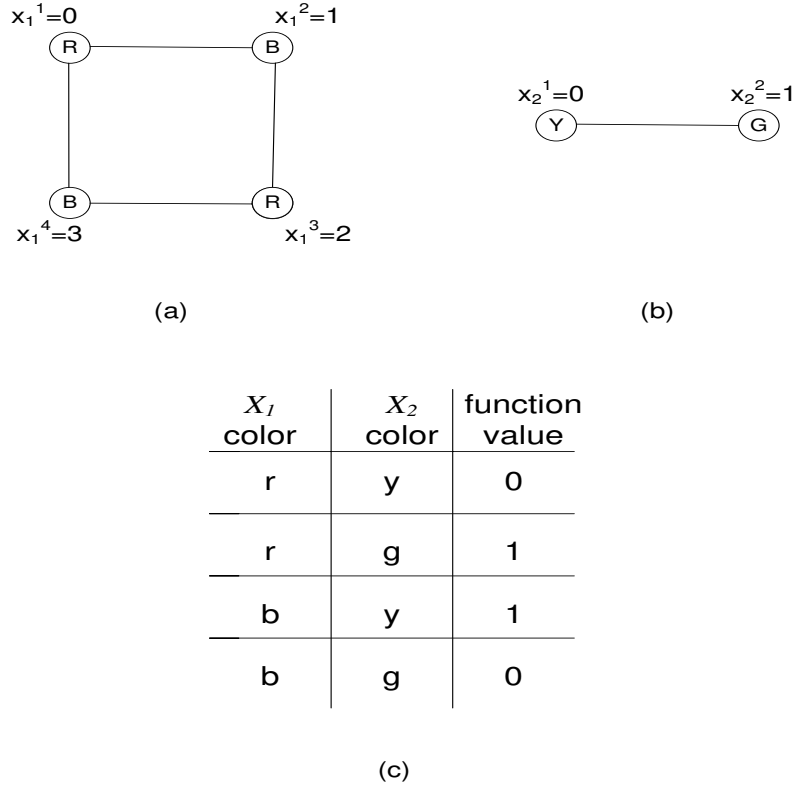}
    \caption{ a) $\gxo$ b) $\gxt$, and c) a decoding look-up table for Example \ref{ex:mod2}. (Different letters written over graph vertices indicate different colors.)}
    \label{fig:look-up-ex}
  \end{figure}

However, this coloring-based coding scheme is not necessarily an achievable scheme. In other words, if we send coloring random variables instead of source random variables, the receiver may not be able to compute its desired function. Hence, we need some conditions to guarantee the achievability of coloring-based coding schemes. We explain this required condition by an example.

\begin{ex}\label{ex:ccc1}

Consider the same network topology as explained in Example \ref{ex:mod2} shown in Figure \ref{fig:3part}-b. Suppose $\cX_1=\{0,1\}$ and $\cX_2=\{0,1\}$. The function values are depicted in Figure \ref{fig:ccc-ex}-a. In particular, $f(0,0)=0$ and $f(1,1)=1$. Dark squares in this figure represent points with zero probability. Figure \ref{fig:ccc-ex}-b demonstrates characteristic graphs of these source random variables. Each has two vertices, not connected to each other. Hence, we can assign them to a same color. Figure \ref{fig:ccc-ex}-b shows these valid colorings for $\gxo$ and $\gxt$. However, one may notice that if we send these coloring random variables instead of source random variables, the receiver would not be able to compute its desired function.
\end{ex}

  \begin{figure}[t]
	\centering
    \includegraphics[width=10.5cm,height=7cm]{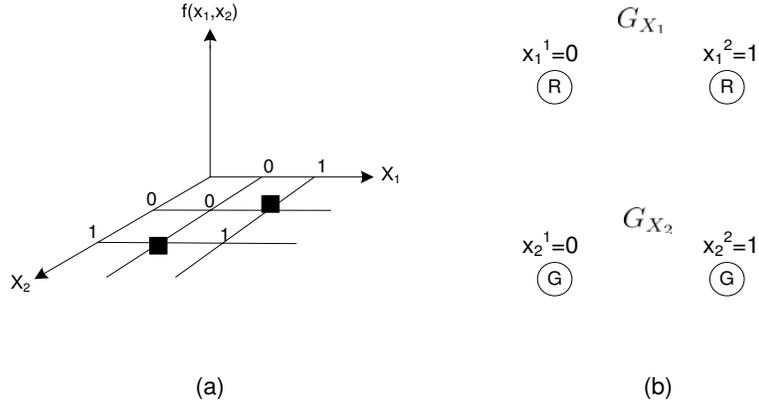}
    \caption{An example for colorings not satisfying C.C.C. (Different letters written over graph vertices indicate different colors.)  }
    \label{fig:ccc-ex}
  \end{figure}

Example \ref{ex:ccc1} demonstrates a case where a coloring-based coding scheme fails to be an achievable scheme. Thus, we need a condition to avoid these situations. We investigate this necessary and sufficient condition in Section \ref{chap:tree}. We call this condition the \textit{coloring connectivity condition} or \textit{C.C.C}. The situation of Example \ref{ex:ccc1} happens when we have a disconnected coloring class (a coloring class is a set of source pairs with the same color for each coordinates). C.C.C. is a condition to avoid this situation. 

Hence, an achievable coding scheme can be expressed as follows. Sources send, instead of source random variables, colorings of their random variables which satisfy C.C.C. Then, they perform source coding on these coloring random variables. Each receiver, by using these colors and a look-up table, can compute its desired function.

However, we may need to consider coloring schemes of conflict graphs of vector extensions of the desired function. In the following, we explain this approach by an example:

\begin{ex}\label{ex:5verex-graph}
Consider the network shown in Figure \ref{fig:3part}-b. Suppose $X_1$ is uniformly distributed over $\cX_1=\{0,1,2,3,4\}$. Consider $X_2$ and $f(X_1,X_2)$ such that we have a graph depicted in Figure \ref{fig:graph-five-vertices} for $\gxo$. Figure \ref{fig:graph-five-vertices} also demonstrates a valid coloring for this graph. Let us call this coloring random variable $\cxo$. Hence, we have $H(\cxo)\approx 1.52$. Now, instead of $X_1$, suppose we encode $X_1\times X_1$ ($X_1^2$), a random variable with $25$ possibilities ($\{00,01,...,44\}$). To make its characteristic graph, we connect two vertices when at least one of their coordinates are connected in $\gxo$. Figure \ref{fig:power-char-graph-ex} illustrates the characteristic graph of $X_1^2$ (referred by $G_{X_1}^2$ and called the second power of the graph $\gxo$). A valid coloring of this graph, called $c_{G_{X_1}^2}$ is shown in this figure. One may notice that we use eight colors to color this graph. We have,
\begin{equation}
\frac{1}{2} H(c_{G_{{X}_1}^2})\approx 1.48 < H(c_{G_{{X}_1}})\approx 1.52.
\end{equation}
\end{ex}

  \begin{figure}[t]
	\centering
    \includegraphics[width=9.5cm,height=7cm]{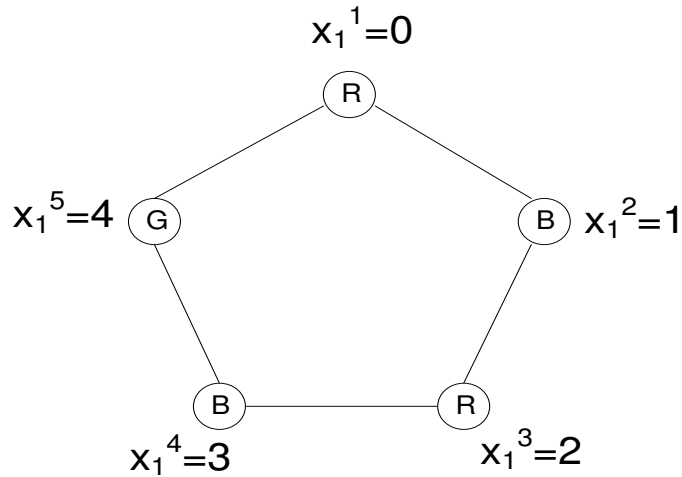}
    \caption{ $\gxo$ described in Example \ref{ex:5verex-graph}. (Different letters written over graph vertices indicate different colors.) }
    \label{fig:graph-five-vertices}
  \end{figure}

Example \ref{ex:5verex-graph} demonstrates this fact that if we assign colors to a sufficiently large power graph of $\gxo$, we can compress source random variables more. In Section \ref{chap:tree}, we show that sending colorings of sufficiently large power graphs of characteristic graphs which satisfy C.C.C. followed by a source coding (such as Slepian-Wolf compression) leads to an achievable coding scheme. On the other hand, any achievable coding scheme for this problem can be viewed as a coloring-based coding scheme satisfying C.C.C. In Section \ref{chap:tree}, we shall explain these concepts in more detail.

  \begin{figure}[t]
	\centering
    \includegraphics[width=8.5cm,height=8.5cm]{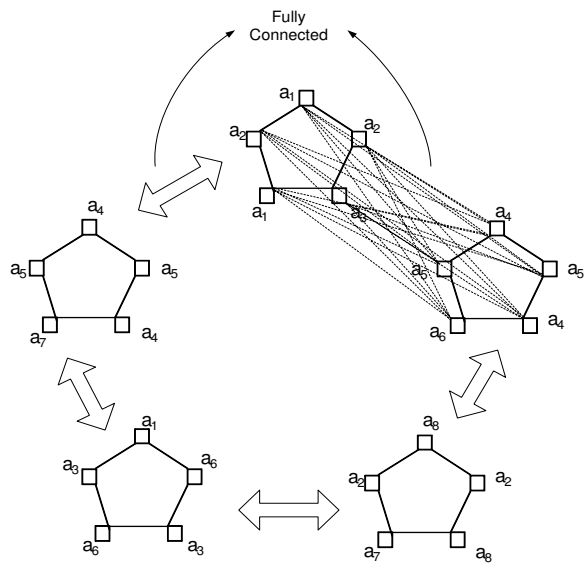}
    \caption{ $G_{{X}_1}^2$, the second power graph of $\gxo$, described in Example \ref{ex:5verex-graph}. Letters $a_1$,...,$a_8$ written over graph vertices indicate different colors. Two subsets of vertices are fully connected if each vertex of one set is connected to every vertex in the other set. }
    \label{fig:power-char-graph-ex}
  \end{figure}

Now, by another example, we explain some issues of the functional compression problem over tree networks.

\begin{ex}\label{ex:tree}
Consider the network topology depicted in Figure \ref{fig:appl-comp}. This is a tree network with four sources, two intermediate nodes and a receiver. Suppose source random variables are independent, with equal probability to be zero or one. In other words, $\mathcal{X}_i=\{0,1\}$ for $i=1,2,3,4$. Suppose the receiver wants to compute a parity check function $f(X_1,X_2,X_3,X_4)=(X_1+X_2+X_3+X_4)\mod 2$. Intermediate nodes are allowed to perform computation.
\end{ex}

In Example \ref{ex:tree}, first notice that characteristic graphs of source random variables are complete graphs. Hence, coloring random variables of sources are equal to source random variables. If intermediate nodes act like relays (i.e., no computations are performed at intermediate nodes), the following set of rates is an achievable scheme:
\begin{eqnarray}
R_{2j}\geq 1 \mbox{ for }1\leq j\leq 4\nonumber\\
R_{1j}\geq 2 \mbox{ for }1\leq j\leq 2
\end{eqnarray}
where $R_{ij}$ are rates of different links depicted in Figure \ref{fig:appl-comp}.

  \begin{figure}[t]
	\centering
    \includegraphics[width=8.5cm,height=6cm]{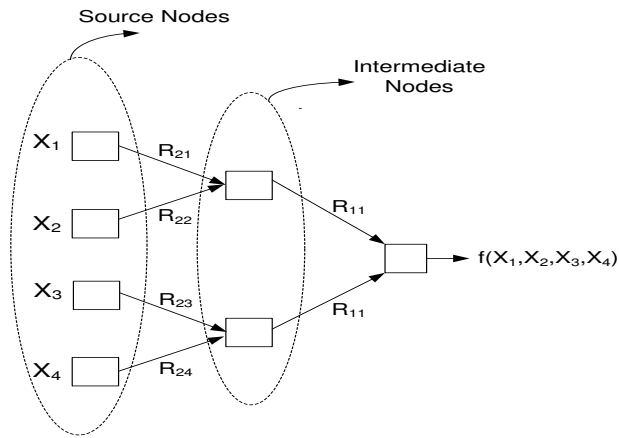}
    \caption{An example of a two stage tree network.}
    \label{fig:appl-comp}
  \end{figure}

However, suppose intermediate nodes perform some computations. Assume source nodes send their coloring random variables satisfying C.C.C. (in this case, they are equal to source random variables because characteristic graphs are complete). Then, each intermediate node makes its own characteristic graph and by using the received colors, picks a corresponding color for its own characteristic graph and sends that color. The receiver, by using the received colors of intermediate nodes' characteristic graphs and a look-up table, can compute its desired function. Figure \ref{fig:look-up-tree} demonstrates this encoding/decoding scheme. For this example, intermediate nodes need to transmit one bit. Therefore, the following set of rates is achievable: 

\begin{equation}
R_{ij}\geq 1.
\end{equation}

for different possible $i$ and $j$. Note that, in Example \ref{ex:tree}, by allowing intermediate nodes to compute, we can reduce transmission rates of some links. This problem is investigated in Section \ref{chap:tree} for a tree network where optimal computation to be performed at intermediate nodes is derived. We also show that for a family of functions and source random variables, intermediate nodes \textit{do not need to perform computation} and acting like relays is an optimal operation for them.  

  \begin{figure}[t]
	\centering
    \includegraphics[width=12cm,height=12cm]{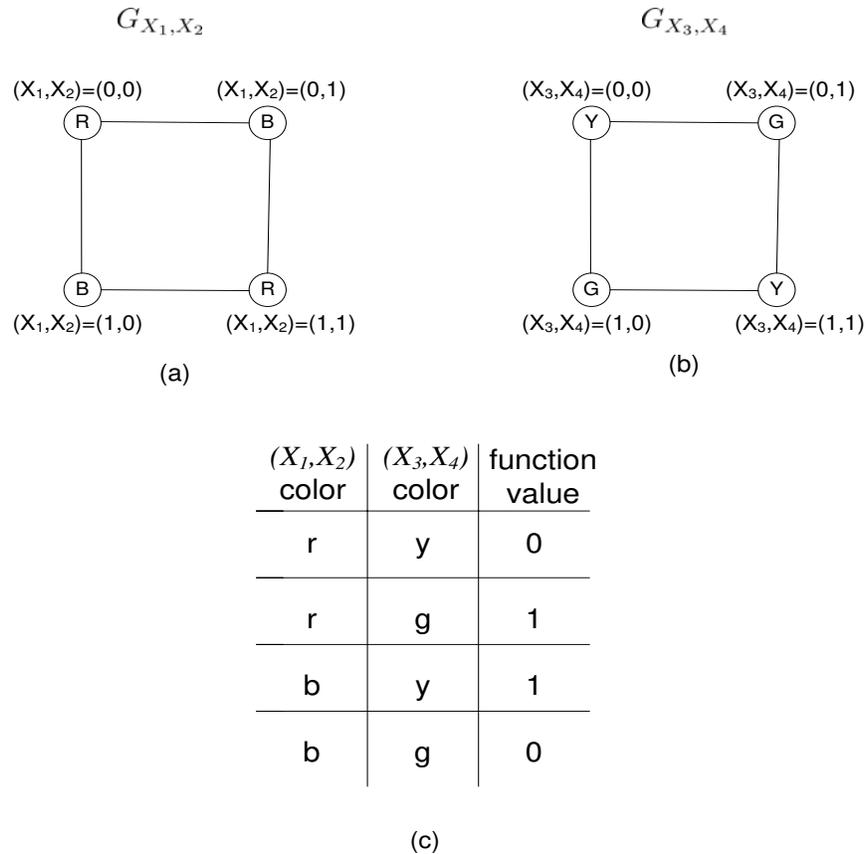}
    \caption{Characteristic graphs and a decoding look-up table for Example \ref{ex:tree}. }
    \label{fig:look-up-tree}
  \end{figure}

The problem of having different desired functions at the receiver with the side information is considered in Section \ref{chap:multi}. For this problem, instead of a characteristic graph, we compute a new graph, called a\textit{ multi-functional characteristic graph}. This graph is basically an OR function of individual characteristic graphs with respect to different functions. In this section, we find a rate region and propose an achievable coding scheme for this problem.

The effect of feedback on the rate-region of functional compression problem is investigated in Section \ref{chap:feedback}. 
If the function at the receiver is the identity function, this problem is the Slepian-Wolf compression with feedback. For this case, having feedback does not give us any gain in terms of the rate. For example, reference \cite{mayak} considers both zero-error and asymptotically zero-error Slepian-Wolf compression with feedback. However, it is not the case when we have a general function $f$ at the receiver. By having feedback, one may outperform rate bounds of the case without feedback.

We consider the problem of distributed functional compression with distortion in Section \ref{chap:distortion}. The objective is to compress correlated discrete sources so that an arbitrary deterministic function of those sources can be computed at the receiver within a distortion level. For this
case, we compute a rate-distortion region and propose an
achievable coding scheme. 

In our proposed coding schemes for different functional compression problems, one needs to compute the minimum entropy coloring (a coloring random variable which minimizes entropy) of a characteristic graph. In general, finding this coloring is an NP-hard problem (\cite{np1}). However, in Section \ref{chap:min-coloring}, we show that, depending on the characteristic graph's structure, there are some interesting cases where finding a minimum entropy coloring is not NP-hard, but tractable and practical. Conclusions and future work are presented in Section \ref{chap:conc}.

\section{Functional Compression Over Tree Networks}\label{chap:tree}

In this section, we consider the problem of functional compression for an arbitrary tree network. Suppose we have $k$ possibly correlated source processes in a tree network, and a receiver at its root wishes to compute a deterministic function of these processes. Other nodes of this tree (called intermediate nodes) are allowed to perform computation to satisfy the node's demand. For this problem, we find a rate region (i.e., feasible rates for different links) of this network when sources are independent and a rate lower bound when sources are correlated. 

The rate region of functional compression problem has been an open problem. However, it has been solved for some simple networks under some special conditions. For instance, \cite{doshi-it} considered a rate region of a network with two transmitters and a receiver under a condition on source random variables. Here, we derive a rate lower bound for an arbitrary tree network based on the graph entropy. We introduce a new condition on colorings of source random variables' characteristic graphs called the coloring connectivity condition (C.C.C.). We show that unlike the condition used in \cite{doshi-it}, this condition is necessary and sufficient for any achievable coding scheme. We also show that, unlike entropy, graph entropy does not satisfy the chain rule. For one stage trees with correlated sources, and general trees with independent sources, we propose a modularized coding scheme based on graph colorings to perform arbitrarily closely to this rate lower bound. We show that in a general tree network case with independent sources, to achieve the rate lower bound, intermediate nodes should perform computation. However, for a family of functions and random variables, which we call chain-rule proper sets, it is sufficient to have intermediate nodes act like relays to perform arbitrarily closely to the rate lower bound.       

In this section, after giving the problem statement and reviewing previous results, we explain our main contributions in this problem.   

\subsection{Problem Setup} \label{subsec:problemsetup}
Consider $k$ discrete memoryless random processes, $\{X_{1}^{i}\}_{i=1}^{\infty}$, ..., $\{X_{k}^{i}\}_{i=1}^{\infty}$, as source processes. Memorylessness is not necessary, and one can approximate a source by a memoryless one with an arbitrary precision \cite{kornerbook}. Suppose these sources are drawn from finite sets $\cX_{1}=\{x_1^{1},x_1^{2},...,x_1^{|\cX_1|}\}$, ..., $\cX_{k}=\{x_k^{1},x_k^{2},...,x_k^{|\cX_k|}\}$. These sources have a joint probability distribution $p(x_{1},...,x_{k})$. We express $n$-sequences of these random variables as $\mathbf{{X_{1}}} = \{X_{1}^{i}\}_{i=l}^{i=l+n-1}$,..., $\mathbf{{X_{k}}} = \{X_{k}^{i}\}_{i=l}^{i=l+n-1}$ with the joint probability distribution $p(\bx_{1},...,\bx_{k})$. Without loss of generality, we assume $l=1$, and to simplify notation, $n$ will be implied by the context if no confusion arises. We refer to the $i^{th}$ element of $\bx_j$ as $x_{ji}$. We use $\bx_j^{1}$, $\bx_j^{2}$,... as different $n$-sequences of $\bX_j$. We shall omit the superscript when no confusion arises. Since the sequence $(\bx_{1},...,\bx_{k})$ is drawn i.i.d. according to $p(x_{1},...,x_{k})$, one can write $p(\bx_{1},...,\bx_{k})=\prod_{i=1}^{n} p(x_{1i},...,x_{ki})$.

Consider a tree network shown in Figure \ref{fig:general-tree}. Suppose we have $k$ source nodes in this network and a receiver in its root. We refer to other nodes of this tree as intermediate nodes. Source node $j$ has an input random process $\{X_{j}^{i}\}_{i=1}^{\infty}$. The receiver wishes to compute a deterministic function $f:\cX_{1}\times...\times\cX_{k}\to\cZ$, or $f:\cX_{1}^{n}\times...\times\cX_{k}^{n}\to\cZ^{n}$, its vector extension.

Note that sources can be at any nodes of the network. However, without loss of generality, we can modify the network by adding some fake leaves to source nodes which are not located in leaves of the network. So, in the achieved network, sources are located in leaves (as an example, look at Figure \ref{fig:leaves}). 

Also, by adding some auxiliary nodes, one can make sources to be in the same distance from the receiver. Hence, we consider source nodes to be in distance $d_{max}$ from the receiver. Consider nodes of a tree with distance $i\geq 1$ from the receiver. We refer to them as the stage $i$ of this tree. Let $w_i$ be the number of such nodes. We refer to the $j^{th}$ node of the $i^{th}$ stage as $n_{ij}$. Its outgoing link is denoted by $e_{ij}$. Suppose this node sends $M_{ij}$ over this edge with a rate $R_{ij}$ (it maps length $n$ blocks of $M_{ij}$, referred to as $\bM_{ij}$, to $\{1,2,...,2^{nR_{ij}}\}$.). 

If this node is a source node (i.e., $n_{d_{max}j}$ for some $j$), then $\bM_{ij}=en_{X_j}(\bX_j)$, where $en_{X_j}$ is the encoding function of the source $j$.

Now, suppose this node is an intermediate node (i.e., $n_{ij}$, $i\notin \{1,d_{max}\}$) with incoming edges $e_{(i+1)1}$, ...,and $e_{(i+1)q}$. We allow this node to compute a function (say $g_{ij}(.)$). Hence, $\bM_{ij}=g_{ij}(\bM_{(i+1)1},...,\bM_{(i+1)q})$.    

The receiver has a decoder $r$ which maps $r:\prod_{1\leq j\leq w_1} \{1,...,2^{n R_{1j}}\}\to\cZ^{n}$. Thus, the receiver computes $r(\bM_{11},...,\bM_{1w_1})=r'(en_{X_1}(\bX_1),...,en_{X_k}(\bX_k))$. We refer to this encoding/decoding scheme as an $n$-distributed functional code. Intermediate nodes are allowed to compute functions, but have no demand of their own. The desired function $f(\bX_1,...,\bX_k)$ at the receiver is the only demand in the network. For any encoding/decoding scheme, the probability of error is defined as 
\begin{equation}
P_{e}^{n}= Pr[(\bx_1,...,\bx_k):{f(\bx_1,...,\bx_k)\neq r'(en_{X_1}(\bx_1),...,en_{X_k}(\bx_k))}].
\end{equation}
A rate tuple of the network is the set of rates of its edges (i.e., $\{R_{ij}\}$ for valid $i$ and $j$). We say a rate tuple is achievable iff there exist a coding scheme operating at these rates so that $P_e^n\to 0$ as $n\to\infty$. The achievable rate region is the set closure of the set of all achievable rates.

\begin{figure}[t]
	\centering
    \includegraphics[width=10.5cm,height=6cm]{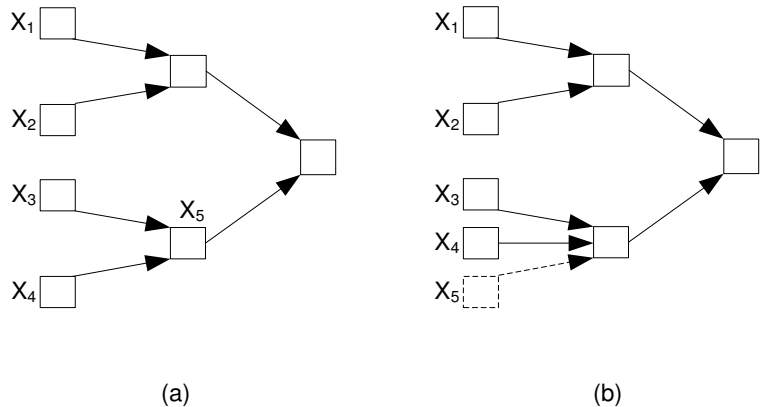}
    \caption{ a) Sources are not necessarily located in leaves b) By adding some fake nodes, one can assume sources are in leaves of the modified tree.}
    \label{fig:leaves}
  \end{figure}

\subsection{Definitions and Prior Results} \label{subsec:doshi}
In this part, first we present some definitions used in formulating our results. We also review some prior results. 

\begin{defn}
The \textit{characteristic graph} $G_{X_1} = (V_{X_1},E_{X_1})$ of $X_1$ with respect to $X_2$, $p(x_1,x_2)$, and function $f(X_1,X_2)$ is defined as follows: $V_{X_1} = \cX_{1}$
and an edge $(x_{1}^{1},x_{1}^{2})\in\mathcal{X}_1^{2}$ is in $E_{X_1}$ iff there exists a $x_{2}^{1}\in\cX_{2}$ such that $p(x_{1}^{1},x_{2}^{1})p(x_{1}^{2},x_{2}^{1})>0$ and $f(x_{1}^{1},x_{2}^{1})\neq f(x_{1}^{2},x_{2}^{1})$. 
\end{defn}

In other words, in order to avoid confusion about the function $f(X_1,X_2)$ at the receiver, if $(x_{1}^{1},x_{1}^{2})\in E_{X_1}$, descriptions of $x_{1}^{1}$ and $x_{1}^{2}$ must be different. 
Shannon first defined this when studying the zero error capacity of noisy channels \cite{s56}. Witsenhausen \cite{w76} used this concept to study a simplified version of our problem where one encodes $X_1$ to compute $f(X_1)$ with zero distortion. The characteristic graph of $X_2$ with respect to $X_1$, $p(x_1,x_2)$, and $f(X_1,X_2)$ is defined analogously and denoted by $G_{X_2}$. One can extend the definition of the characteristic graph to the case of having more than two random variables. Suppose $X_1$, ..., $X_k$ are $k$ random variables defined in Section \ref{subsec:problemsetup}.

\begin{defn}
The \textit{characteristic graph} $G_{X_1} = (V_{X_1},E_{X_1})$ of $X_1$ with respect to random variables $X_2$,...,$X_k$, $p(x_1,...,x_k)$, and $f(X_1,...,X_k)$ is defined as follows: $V_{X_1} = \mathcal{X}_1$
and an edge $(x_{1}^{1},x_{1}^{2})\in\mathcal{X}_{1}^{2}$ is in $E_{x_1}$ if there exist $x_{j}^{1}\in\mathcal{X}_j$ for $2 \leq j\leq k$ such that $p(x_{1}^{1},x_{2}^{1},...,x_{k}^{1})p(x_{1}^{2},x_{2}^{1},...,x_{k}^{1})>0$ and $f(x_{1}^{1},x_{2}^{1},...,x_{k}^{1})\neq f(x_{1}^{2},x_{2}^{1},...,x_{k}^{1})$. 
\end{defn}


\begin{ex}\label{ex:2-mod2}
To illustrate the idea of confusability and the characteristic graph, consider two random variables $X_1$ and $X_2$ such that $\cX_1=\{0,1,2,3\}$ and $\cX_2=\{0,1\}$ where they are uniformly and independently distributed on their own supports. Suppose $f(X_1,X_2)=(X_1+X_2)$ mod $2$ is to perfectly reconstructed at the receiver. Then, the characteristic graph of $X_1$ with respect to $X_2$, $p(x_1,x_2)=\frac{1}{8}$, and $f$ is shown in Figure \ref{fig:char-graph-sq}-a.    
\end{ex}

The following definition can be found in \cite{k73}.

\begin{defn}\label{def:korner}
Given a graph $G_{X_1}=(V_{X_1},E_{X_1})$ and a distribution on its vertices $V_{X_1}$, graph entropy is
\begin{align}
H_{\gxo}(X_1) &= \min_{X_1\in W_1\in \Gamma(G_{X_1})}I(X_1;W_1),
\end{align}
where $\Gamma(G_{X_1})$ is the set of all maximal independent 
sets of $\gxo$.  
\end{defn}

The notation $X_1\in W_1\in\Gamma(\gxo)$ means that we are minimizing over all distributions $p(w_1,x_1)$ such that $p(w_1,x_1)>0$ implies $x_1\in w_1$, where $w_1$ is a maximal independent set of the graph $G_{x_1}$. 



\begin{ex}
Consider the scenario described in Example \ref{ex:2-mod2}. For the characteristic graph of $X_1$ shown in Figure \ref{fig:char-graph-sq}-a, the set of maximal independent sets is $W_1=\{\{0,2\},\{1,3\}\}$. To minimize $I(X_1;W_1)=H(X_1)-H(X_1|W_1)=\log(4)-H(X_1|W_1)$, one should maximize $H(X_1|W_1)$. Because of the symmetry of the problem, to maximize $H(X_1|W_1)$, $p(w_1)$ must be uniform over two possible maximal independent sets of $\gxo$. Since each maximal independent set $w_1\in W_1$ has two $X_1$ values, thus, $H(X_1|w_1)=\log(2)$ bit, and since $p(w_1)$ is uniform, $H(X_1|W_1)=\log(2)$ bit. Therefore, $H_{\gxo}(X_1)=\log(4)-\log(2)=1$ bit. One can see if we want to encode $X_1$ ignoring the effect of the function $f$, we need $H(X_1)=\log(4)=2$ bits. We will show that, for this example, functional compression saves us $1$ bit in every $2$ bits compared to the traditional data compression.   
\end{ex}

Witsenhausen \cite{w76} showed that the graph entropy is the minimum rate at which a single source can be encoded so that a function of that source can be computed with zero distortion. Orlitsky and Roche \cite{or01} defined an extension of K\"orner's graph entropy, the \textit{conditional graph entropy}.

\begin{defn}\label{def:or}
The conditional graph entropy is 
\begin{align}
H_{G_{X_1}}(X_1|X_2)  &= \min_{\substack{
  X_1\in W_1\in \Gamma(\gxo)\\
  W_1-X_1-X_2}}
I(W_1;X_1|X_2).
\end{align}
\end{defn}

Notation $W_1-X_1-X_2$ indicates a Markov chain. If $X_1$ and $X_2$ are independent, $H_{G_{X_1}}(X_1|X_2)=H_{G_{X_1}}(X_1)$. 
To illustrate this concept, let us consider an example borrowed from \cite{or01}.


\begin{ex}
When $f(X_1,X_2)=X_1$, $H_{\gxo}(X_1|X_2)=H(X_1|X_2)$. 

To show this, consider the characteristic graph of $X_1$, denoted as $\gxo$. Since $f(X_1,X_2)=X_1$, then for every $x_2^{1}\in\cX_2$, the set $\{x_1^{i}:p(x_1^{i},x_2^{1})>0\}$ of possible $x_1^{i}$ are connected to each other (i.e., this set is a clique of $\gxo$). Since the intersection of a clique and a maximal independent set is a singleton, $X_2$ and the maximal independent set $W_1$ containing $X_1$ determine $X_1$. So,

\begin{eqnarray}
H_{G_{X_1}}(X_1|X_2)  &=& \min_{\substack{
  X_1\in W_1\in \Gamma(\gxo)\\
  W_1-X_1-X_2}}
I(W_1;X_1|X_2)\nonumber\\
&=&H(X_1|X_2)-\max_{\substack{
  X_1\in W_1\in \Gamma(\gxo)}}
H(X_1|W_1,X_2)\\
&=&H(X_1|X_2).\nonumber
\end{eqnarray}    
\end{ex}


\begin{defn}
A vertex coloring of a graph is a function $\cxo(X_1):V_{x_1} \to \mathbbm{N} $ of a graph $\gxo= (V_{X_1},E_{X_1})$ such that $(x_{1}^{1},x_{1}^{2})\in E_{X_1}$ implies $\cxo(x_{1}^{1})\neq \cxo(x_{1}^{2})$. The entropy of a coloring is the entropy of the induced distribution on colors. Here, $p(\cxo(x_{1}^{i}))=p(\cxo^{-1}(\cxo(x_{1}^{i})))$, where $\cxo^{-1}(\cxo(x_{1}^{i})) = \{x_{1}^{j} : \cxo(x_{1}^{j})=\cxo(x_{1}^{i})\}$ for all valid $j$. This subset of vertices with the same color is called a \textit{color class}. We refer to a coloring which minimizes the entropy as a minimum entropy coloring. We use $\Cxo$ as the set of all valid colorings of a graph $\gxo$. 
\end{defn}


\begin{ex}
Consider again the random variable $X_1$ described in Example \ref{ex:2-mod2}, whose characteristic graph $\gxo$ is shown in Figure \ref{fig:char-graph-sq}-a. A valid coloring for $\gxo$ is shown in Figure $\ref{fig:look-up-ex}$-a. One can see that, in this coloring, two connected vertices are assigned to different colors. Specifically, $\cxo(X_1)=\{r,b\}$. So, $p(\cxo(x_1^{i})=r)=p(x_1^{i}=0)+p(x_1^{i}=2)$, and $p(\cxo(x_1^{i})=b)=p(x_1^{i}=1)+p(x_1^{i}=3)$.       
\end{ex}

We define a power graph of a characteristic graph as follows:

\begin{defn}
The $n$-th power of a graph $\gxo$ is a graph $\gxon=(\vxon,\exon)$ such that $\vxon=\cX_1^{n}$ and $(\bx_{1}^{1},\bx_{1}^{2})\in\exon$ when there exists at least one $i$ such that $(x_{1i}^{1},x_{1i}^{2})\in\exo$. We denote a valid coloring of $\gxon$ by $\cxon(\bX_1)$.    
\end{defn}

One may ignore atypical sequences in a sufficiently large power graph of a conflict graph and then, color that graph. This coloring is called an $\eps$-coloring of a graph and is defined as follows:

\begin{defn}
Given a non-empty set $\cA \subset\cX_1\times\cX_2$, define $\ph(x_1,x_2) = p(x_1,x_2)/p(\cA)$ when $(x_1,x_2)\in\cA$, and $\ph(x,y)=0$ otherwise. $\ph$ is the distribution over $(x_1,x_2)$ conditioned on $(x_1,x_2)\in\cA$. Denote the characteristic graph of $X_1$ with respect to $X_2$, $\ph(x_1,x_2)$, and $f(X_1,X_2)$ as $\Gh_{X_1}=(\Vh_{X_1},\Eh_{X_1})$ and the characteristic graph of $X_2$ with respect to $X_1$, $\ph(x_1,x_2)$, and $f(X_1,X_2)$ as $\Gh_{X_2}=(\Vh_{X_2},\Eh_{X_2})$. Note that $\Eh_{X_1}\subseteq E_{X_1}$ and $\Eh_{X_2}\subseteq E_{X_2}$. Suppose $p(\cA)\geq 1-\epsilon$. We say that $\cxo(X_1)$ and $\cxt(X_2)$ are $\epsilon$-colorings of $G_{X_1}$ and $G_{x_2}$
if they are valid colorings of $\Gh_{X_1}$ and $\Gh_{X_2}$. 
\end{defn}

In \cite{ao96}, the \textit{Chromatic entropy} of a graph $\gxo$ is defined as
\begin{defn}
\begin{align*} 
H_{\gxo}^{\chi}(X_1) &= \min_{\cxo\text{ is an }\epsilon\text{-coloring of }\gxo} H(\cxo(X_1)).
\end{align*}
\end{defn}

The chromatic entropy is a representation of the chromatic number of high probability subgraphs of the characteristic graph. 
In \cite{doshi-it}, the conditional chromatic entropy is defined as

\begin{defn}
\begin{align*}
H_{\gxo}^{\chi}(X_1|X_2) &= \min_{\cxo\text{ is an }\epsilon\text{-coloring of }\gxo} H(\cxo(X_1)|X_2).
\end{align*}
\end{defn}

Regardless of $\epsilon$, the above optimizations are minima, rather than infima, because there are finitely many subgraphs of any fixed graph $\gxo$, and therefore there are only finitely many $\epsilon$-colorings, regardless of $\epsilon$. 

In general, these optimizations are NP-hard (\cite{np1}). But, depending on the desired function $f$, there are some interesting cases that they are not NP-hard. We discuss these cases in Section \ref{chap:min-coloring}. 

K\"orner showed in \cite{k73} that, in the limit of large $n$, there is a relation between the chromatic entropy and the graph entropy.

\begin{thm} \label{def:eq-uncond}
\begin{align}
\limtoi \frac1n H_{\gxon}^{\chi}(\bX_1) = H_{\gxo}(X_1).
\end{align}
\end{thm}

This theorem implies that the receiver can asymptotically compute a deterministic function of a discrete memoryless source, by first coloring a sufficiently large power of the characteristic graph of the source random variable with respect to the function, and then, encoding achieved colors using any encoding scheme which achieves the entropy bound of the coloring RV. In the previous approach, to achieve the encoding rate close to graph entropy of $X_1$, one should find the optimal distribution over the set of maximal independent sets of $\gxo$. But, this theorem allows us to find the optimal coloring of $\gxon$, instead of the optimal distribution on maximal independent sets. One can see that this approach modularizes the encoding scheme into two parts, a graph coloring module, followed by a Slepian-Wolf compression module.      

The conditional version of the above theorem is proven in \cite{doshi-it}.

\begin{thm} \label{def:eq-cond}
\begin{align}
\limtoi \frac1n H_{\gxon}^{\chi}(\bX_1 |\bX_2) = H_{\gxo}(X_1 | X_2).
\end{align}
\end{thm}

This theorem implies a practical encoding scheme for the problem of functional compression with side information where the receiver wishes to compute $f(X_1,X_2)$, when $X_2$ is available at the receiver as the side information. Orlitsky and Roche showed in \cite{or01} that $H_{\gxo}(X_1|X_2)$ is the minimum achievable rate for this problem. Their proof uses random coding arguments and shows the existence of an optimal coding scheme. This theorem presents a modularized encoding scheme where one first finds the minimum entropy coloring of $\gxon$ for large enough $n$, and then uses a compression scheme on the coloring random variable (such as Slepian-Wolf \cite{sw73}) to achieve a rate arbitrarily close to $H(\cxon(\bX_1)|\bX_2)$. This encoding scheme guarantees computation of the function at the receiver with a vanishing probability of error.

All these results considered only functional compression with side information at the receiver (Figure \ref{fig:3part}-a). Consider the network shown in Figure \ref{fig:3part}-b. It shows a network with two source nodes and a receiver which wishes to compute a function of the sources' values. In general, the rate-region of this network has not been determined. However, \cite{doshi-it} determined a rate-region of this network when source random variables satisfy a condition called the zigzag condition, defined below.

We refer to the $\eps$-joint-typical set of sequences of random variables $\bX_1$, ..., $\bX_k$ as $\Ten$. $k$ is implied in this notation for simplicity. $\Ten$ can be considered as a strong or weak typical set (\cite{kornerbook}).     

\begin{defn}\label{def:zigzag}
A discrete memoryless source $\{(X_{1}^{i},X_{2}^{i})\}_{i\in \mathbbm{N}}$ with a distribution $p(x_1,x_2)$ satisfies the zigzag condition if for any $\epsilon$ and some $n$, $(\bx_{1}^{1},\bx_{2}^{1})$, $(\bx_{1}^{2},\bx_{2}^{2})\in\Ten$, there exists some $(\bx_{1}^{3},\bx_{2}^{3})\in\Ten$ such that $(\bx_{1}^{3},\bx_{2}^{i}),(\bx_{1}^{i},\bx_{2}^{3})\in \Teen$ for each $i\in\{1,2\}$, and $(x_{1j}^{3},x_{2j}^{3})=(x_{1j}^{i},x_{2j}^{3-i})$ for some $i \in \{1,2\}$ for each $j$.
\end{defn}

 In fact, the zigzag condition forces many source sequences to be typical. We first explain the  results of \cite{doshi-it}. Then, in Section \ref{sec:one-stage}, we compute a rate-region without the need for the zigzag condition. Then, we extend our results to the case of having $k$ source nodes. 

Reference \cite{doshi-it} shows that, if the source random variables satisfy the zigzag condition, an achievable rate region for this network is the set of all rates that can be achieved through graph colorings. The zigzag condition is a restrictive condition which does not depend on the desired function at the receiver. This condition is not necessary, but sufficient.

\subsection{A Rate Region for One-Stage Tree Networks}\label{sec:one-stage}

In this section, we want to find a rate region for a general one stage tree network without having any restrictive conditions such as the zigzag condition. Consider the network shown in Figure \ref{fig:one-stage} with $k$ sources.

\begin{defn}\label{def:path}
 A path with length $m$ between two points $Z_1=(x_1^{1},x_2^{1},...,x_{k}^{1})$, and $Z_m=(x_1^{2},x_2^{2},...,x_{k}^{2})$ is determined by $m-1$ points $Z_i$, $1\leq i\leq m$ such that,

i) $P(Z_i)>0$, for all $1\leq i\leq m$.

ii) $Z_i$ and $Z_{i+1}$ only differ in one of their coordinates.   
\end{defn}

Definition \ref{def:path} can be generalized to two $n$-length vectors as follows. 

\begin{defn}
 A path with length $m$ between two points $\bZ_1=(\bx_1^{1},\bx_2^{1},...,\bx_{k}^{1})\in\Ten$, and $\bZ_m=(\bx_1^{2},\bx_2^{2},...,\bx_{k}^{2})\in\Ten$ are determined by $m-1$ points $\bZ_i$, $1\leq i\leq m$ such that,

i) $\bZ_i\in\Ten$, for all $1\leq i\leq m$.

ii) $\bZ_i$ and $\bZ_{i+1}$ only differ in one of their coordinates.   
\end{defn}

Note that, each coordinate of $\bZ_i$ is a vector with length $n$.

\begin{defn}\label{def:jc}
A joint-coloring family $J_C$ for random variables $X_1$, ..., $X_k$ with characteristic graphs $\gxo$,...,$\gxk$, and any valid colorings $\cxo$,...,$\cxk$, respectively is defined as $J_C=\{j_c^1,...,j_c^{n_{j_c}}\}$ where $j_c^i$ is the collection of points $(x_{1}^{i_{1}},x_{2}^{i_{2}},...,x_{k}^{i_{k}})$ whose coordinates have the same color (i.e., $j_c^i=\big\{(x_{1}^{i_{1}},x_{2}^{i_{2}},...,x_{k}^{i_{k}}),(x_{1}^{l_{1}},x_{2}^{l_{2}},...,x_{k}^{l_{k}}):\cxo(x_{1}^{i_{1}})= \cxo(x_{1}^{l_{1}})\mbox{ , ..., }\cxk(x_{k}^{i_{k}})=\cxk(x_{k}^{l_{k}})\big\}$, for any valid $i_1$,...$i_k$, and $l_1$,...$l_k$). Each $j_c^i$ is called a joint coloring class where $n_{j_c}$ is the number of joint coloring classes of a joint coloring family.
\end{defn}

We say a joint coloring class $j_c^i$ is connected if between any two points in $j_c^i$, there exists a path that lies in $j_c^i$. Otherwise, it is disconnected. Definition \ref{def:jc} can be expressed for random vectors $\bX_1$,...,$\bX_k$ with characteristic graphs $\gxon$,...,$\gxkn$, and any valid $\epsilon$-colorings $\cxon$,...,$\cxkn$, respectively. 

\begin{defn}
Consider random variables $X_1$, ..., $X_k$ with characteristic graphs $\gxo$, ..., $\gxk$, and any valid colorings $\cxo$, ..., $\cxk$. We say a joint coloring class $j_c^i\in J_C$ satisfies the \textit{Coloring Connectivity Condition} (C.C.C.) when it is connected or disconnected parts of $j_c^i$ have the same function value. We say colorings $\cxo$, ..., $\cxk$ satisfy C.C.C. when  all joint coloring classes satisfy C.C.C. 
\end{defn}

C.C.C. can be expressed for random vectors $\bX_1$, ..., $\bX_k$ with characteristic graphs $\gxon$, ..., $\gxkn$, and any valid $\epsilon$-colorings $\cxon$, ..., $\cxkn$, respectively. 

\begin{ex}
For example, suppose we have two random variables $X_1$ and $X_2$ with characteristic graphs $\gxo$ and $\gxt$. Let us assume $\cxo$ and $\cxt$ are two valid colorings of $\gxo$ and $\gxt$, respectively. Assume $\cxo(x_1^{1})=\cxo(x_1^{2})$ and $\cxt(x_2^{1})=\cxt(x_2^{2})$. Suppose $j_c^1$ represents this joint coloring class. In other words, $j_c^1=\{(x_1^{i},x_2^{j})\}$, for all $1 \leq i,j\leq 2$ when $p(x_1^{i},x_2^{j})> 0$. Figure \ref{fig:ccc-ab} considers two different cases. The first case is when $p(x_1^{1},x_2^{2})=0$, and other points have a non-zero probability. It is illustrated in Figure \ref{fig:ccc-ab}-a. One can see that there exists a path between any two points in this joint coloring class. So, this joint coloring class satisfies C.C.C. If other joint coloring classes of $\cxo$ and $\cxt$ satisfy C.C.C., we say $\cxo$ and $\cxt$ satisfy C.C.C. Now, consider the second case depicted in Figure \ref{fig:ccc-ab}-b. In this case, we have $p(x_1^{1},x_2^{2})=0$, $p(x_1^{2},x_2^{1})=0$, and other points have a non-zero probability. One can see that there is no path between $(x_1^{1},x_2^{1})$ and $(x_1^{2},x_2^{2})$ in $j_c^1$. So, though these two points belong to a same joint coloring class, their corresponding function values can be different from each other. Thus, $j_c^1$ does not satisfy C.C.C. for this example. Therefore, $\cxo$ and $\cxt$ do not satisfy C.C.C. 
\end{ex}

 \begin{figure}[t]
	\centering
    \includegraphics[width=10.5cm,height=6cm]{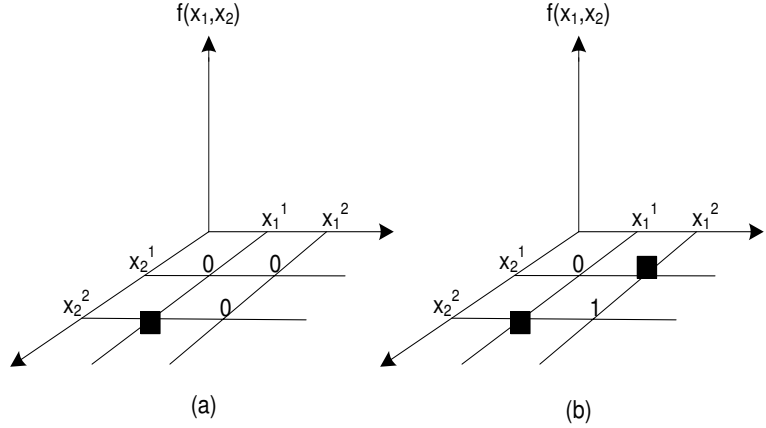}
    \caption{ Two examples of a joint coloring class: a) satisfying C.C.C. b) not satisfying C.C.C. Dark squares indicate points with zero probability. Function values are depicted in the picture.}
    \label{fig:ccc-ab}
  \end{figure}

\begin{lem}\label{lem:triv}
Consider two random variables $X_1$ and $X_2$ with characteristic graphs $\gxo$ and $\gxt$ and any valid colorings $\cxo(X_1)$ and $\cxt(X_2)$ respectively, where $\cxt(X_2)$ is a trivial coloring, assigning different colors to different vertices (to simplify the notation, we use $\cxt(X_2)=X_2$ to refer to this coloring). These colorings satisfy C.C.C. Also, $\cxon(\bX_1)$ and $\cxtn(\bx_2)=\bX_2$ satisfy C.C.C., for any $n$. 
\end{lem}

\begin{proof}First, we know that any random variable $X_2$ by itself is a trivial coloring of $\gxt$ such that each vertex of $\gxt$ is assigned to a different color. So, $J_C$ for $\cxo(X_1)$ and $\cxt(X_2)=X_2$ can be written as $J_C=\{j_c^1,...,j_c^{n_{j_c}}\}$ such that $j_c^1=\{(x_1^{i},x_2^{1}): \cxo(x_1^{i})=\sigma_i\}$, where $\sigma_i$ is a generic color. Any two points in $j_c^1$ are  connected to each other with a path with length one. So, $j_c^1$ satisfies C.C.C. This arguments hold for any $j_c^i$ for any valid $i$. Thus, all joint coloring classes and therefore, $\cxo(X_1)$ and $\cxt(X_2)=X_2$ satisfy C.C.C. The argument for $\cxon(\bX_1)$ and $\cxtn(\bX_2)=\bX_2$ is similar.      
\end{proof}

\begin{lem}\label{lem:equal}
Consider random variables $X_1$, ..., $X_k$ with characteristic graphs $\gxo$, ..., $\gxk$, and any valid colorings $\cxo$, ..., $\cxk$ with joint coloring class $J_C=\{j_c^i:i\}$. For any two points $(x_1^{1},...,x_k^{1})$ and $(x_1^{2},...,x_k^{2})$ in $j_c^i$, $f(x_1^{1},...,x_k^{1})=f(x_1^{2},...,x_k^{2})$ if and only if $j_c^i$ satisfies C.C.C. 
\end{lem}

\begin{proof}
We first show that if $j_c^i$ satisfies C.C.C., then, for any two points $(x_1^{1},...,x_k^{1})$ and $(x_1^{2},...,x_k^{2})$ in $j_c^i$, $f(x_1^{1},...,x_k^{1})=f(x_1^{2},...,x_k^{2})$ . Since $j_c^i$ satisfies C.C.C., either $f(x_1^{1},...,x_k^{1})=f(x_1^{2},...,x_k^{2})$, or there exists a path with length $m-1$ between these two points $Z_1=(x_1^{1},...,x_k^{1})$ and $Z_m=(x_1^{2},...,x_k^{2})$, for some $m$. Two consecutive points $Z_j$ and $Z_{j+1}$ in this path, differ in just one of their coordinates. Without loss of generality, suppose they differ in their first coordinate. In other words, suppose $Z_j=(x_1^{j_1},x_2^{j_2}...,x_k^{j_k})$ and $Z_{j+1}=(x_1^{j_0},x_2^{j_2}...,x_k^{j_k})$. Since these two points belong to $j_c^i$, $\cxo(x_1^{j_1})=\cxo(x_1^{j_0})$. If $f(Z_j)\neq f(Z_{j+1})$, there would exist an edge between $x_1^{j_1}$ and $x_1^{j_0}$ in $\gxo$ and they could not have the same color. So, $f(Z_j)=f(Z_{j+1})$. By applying the same argument inductively for all two consecutive points in the path between $Z_1$ and $Z_m$, one can get $f(Z_1)=f(Z_2)=...=f(Z_m)$.  

If $j_c^i$ does not satisfy C.C.C., it means that there exists at least two points $Z_1$ and $Z_2$ in $j_c^i$ such that no path exists between them with different function values. As an example, consider Figure \ref{fig:ccc-ab}-b. Hence, the function value is the same in a joint coloring class iff it satisfies C.C.C.         
\end{proof}

\begin{lem}\label{lem:same-value}
Consider random variables $\bX_1$, ..., $\bX_k$ with characteristic graphs $\gxon$, ..., $\gxkn$, and any valid $\epsilon$-colorings $\cxon$, ..., $\cxkn$ with the joint coloring class $J_C=\{j_c^i:i\}$. For any two points $(\bx_1^{1},...,\bx_k^{1})$ and $(\bx_1^{2},...,\bx_k^{2})$ in $j_c^i$, $f(\bx_1^{1},...,\bx_k^{1})= f(\bx_1^{2},...,\bx_k^{2})$ if and only if $j_c^i$ satisfies C.C.C. 
\end{lem}

\begin{proof}
The proof is similar to Lemma \ref{lem:equal}. The only difference is to use the definition of C.C.C. for $\cxon$, ..., $\cxkn$. Since $j_c^i$ satisfies C.C.C., either $f(\bx_1^{1},...,\bx_k^{1})= f(\bx_1^{2},...,\bx_k^{2})$, or there exists a path with length $m-1$ between any two points $\bZ_1=(\bx_1^{1},...,\bx_k^{1})\in\Ten$ and $\bZ_m=(\bx_1^{2},...,\bx_k^{2})\in\Ten$ in $j_c^i$, for some $m$. Consider two consecutive points $\bZ_j$ and $\bZ_{j+1}$ in this path. They differ in one of their coordinates (suppose they differ in their first coordinate). In other words, suppose $\bZ_j=(\bx_1^{j_1},\bx_2^{j_2}...,\bx_k^{j_k})\in\Ten$ and $\bZ_{j+1}=(\bx_1^{j_0},\bx_2^{j_2}...,\bx_k^{j_k})\in\Ten$. Since these two points belong to $j_c^i$, $\cxo(\bx_1^{j_1})=\cxo(\bx_1^{j_0})$. If $f(\bZ_j)\neq f(\bZ_{j+1})$, there would exist an edge between $\bx_1^{j_1}$ and $\bx_1^{j_0}$ in $\gxon$ and they could not get the same color. Thus, $f(\bZ_j)=f(\bZ_{j+1})$. By applying the same argument for all two consecutive points in the path between $\bZ_1$ and $\bZ_m$, one can get $f(\bZ_1)=f(\bZ_2)=...=f(\bZ_m)$. The converse part is similar to Lemma \ref{lem:equal}.  
\end{proof}   

Next, we want to show that, if $X_1$ and $X_2$ satisfy the zigzag condition given in Definition \ref{def:zigzag}, any valid colorings of their characteristic graphs satisfy C.C.C., but not vice versa. In other words, we want to show that the zigzag condition used in \cite{doshi-it} is sufficient but not necessary. 
 
\begin{lem}
If two random variables $X_1$ and $X_2$ with characteristic graphs $\gxo$ and $\gxt$ satisfy the zigzag condition, any valid colorings $\cxo$ and $\cxt$ of $\gxo$ and $\gxt$ satisfy C.C.C., but not vice versa.  
\end{lem}

\begin{proof}
Suppose $X_1$ and $X_2$ satisfy the zigzag condition, and $\cxo$ and $\cxt$ are two valid colorings of $\gxo$ and $\gxt$, respectively. We want to show that these colorings satisfy C.C.C. To do this, consider two points $(x_1^{1},x_2^{1})$ and $(x_1^{2},x_2^{2})$ in a joint coloring class $j_c^i$. The definition of the zigzag condition guarantees the existence of a path with length two between these two point. Thus, $\cxo$ and $\cxt$ satisfy C.C.C. 

The second part of this Lemma says that the converse part is not true. To have an example, one can see that in a special case considered in Lemma \ref{lem:triv}, those colorings always satisfy C.C.C.  without having any condition such as the zigzag condition.    
\end{proof}

\begin{defn} \label{def:joint} 
For random variables $X_1$, ..., $X_k$ with characteristic graphs $\gxo$, ..., $\gxk$, the joint graph entropy is defined as follows:
\begin{equation} 
H_{\sgjk}(X_1,...,X_k) \triangleq \limtoi \min_{\cxon,...,\cxkn} \frac{1}{n} H(\cxon(\bX_1),...,\cxkn(\bX_k))
\end{equation}
\end{defn}

in which $\cxon(\bX_1)$, ..., $\cxkn(\bX_k)$ are $\epsilon$-colorings of $\gxon$, ..., $\gxkn$ satisfying C.C.C. We refer to this joint graph entropy as $H_{\sgs}$ where $S=\{1,2,...,k\}$. Note that, this limit exists because we have a monotonically decreasing sequence bounded below.
Similarly, we can define the conditional graph entropy.
\begin{defn}\label{def:cond}
For random variables $X_1$, ..., $X_k$ with characteristic graphs $\gxo$, ..., $\gxk$, the conditional graph entropy can be defined as follows:
\begin{eqnarray} 
&&H_{\sgji}(X_1,...,X_i|X_{i+1},...,X_{k}) \nonumber\\
&&\triangleq \limtoi \min\frac{1}{n} H(\cxon(\bX_1),...,c_{G_{\bX_i}^{n}}(\bX_i)|c_{G_{\bX_{i+1}}^{n}}(\bX_{i+1}),...,\cxkn(\bX_k))
\end{eqnarray}
where the minimization is over $\cxon(\bX_1)$, ..., $\cxkn(\bX_k)$, which are $\epsilon$-colorings of $\gxon$, ..., $\gxkn$ satisfying C.C.C. 
\end{defn}

\begin{lem}
For $k=2$, Definitions \ref{def:or} and \ref{def:cond} are the same. 
\end{lem}
\begin{proof}
By using the data processing inequality, we have
\begin{eqnarray} 
H_{\gxo}(X_1|X_2) &=& \limtoi \min_{\cxon,\cxtn} \frac{1}{n} H(\cxon(\bX_1)|\cxtn(\bX_2))\nonumber\\
 &=& \limtoi \min_{\cxon} \frac{1}{n} H(\cxon(\bX_1)|\bX_2).\nonumber
\end{eqnarray}

Then, Lemma \ref{lem:triv} implies that $\cxon(\bX_1)$ and $\cxtn(\bx_2)=\bX_2$ satisfy C.C.C. A direct application of Theorem \ref{def:eq-cond} completes the proof. 
\end{proof}

Note that, by this definition, \textit{the graph entropy does not satisfy the chain rule}.  

Suppose $\cS(k)$ denotes the power set of the set $\{1,2,...,k\}$ excluding the empty subset. Then, for any $S\in\cS(k)$,
\[ X_{S}\triangleq \{X_i:i\in S\}.\] Let $S^{c}$ denote the complement of $S$ in $\cS(k)$. For $S=\{1,2,...,k\}$, denote $S^{c}$ as the empty set.  To simplify notation, we refer to a subset of sources by $X_S$. For instance, $\cS(2)=\{\{1\},\{2\},\{1,2\}\}$, and for $S=\{1,2\}$, we write $H_{\sgs}(X_s)$ instead of $H_{\sgjt}(X_1,X_2)$.

\begin{thm}\label{th:gendoshi}
A rate region of the network shown in Figure \ref{fig:one-stage} is determined by these conditions:
\begin{equation}\label{eq:doshi}
\forall S\in\cS(k) \Longrightarrow \sum_{i\in S} R_{1i} \geq H_{\sgs} (X_{S}|X_{S^{c}}).\\
\end{equation}
\end{thm}

\begin{proof}
We first show the achievability of this rate region. We also propose a modularized encoding/decoding scheme in this part. Then, for the converse, we show that no encoding/decoding scheme can outperform this rate region.

1)\textit{Achievability}: 

\begin{lem}\label{lem:fhat}
Consider random variables $\bX_1$, ..., $\bX_k$ with characteristic graphs $\gxon$, ..., $\gxkn$, and any valid $\epsilon$-colorings $\cxon$, ..., $\cxkn$ satisfying C.C.C., for sufficiently large $n$. There exists 
\begin{equation}
\hat{f}: \cxon(\cX_1)\times...\times\cxkn(\cX_k)\to\cZ^{n}
\end{equation}
such that $\hat{f}(\cxon(\bx_1),...,\cxkn(\bx_k))=f(\bx_1,...,\bx_k)$, for all $(\bx_1,...,\bx_k)\in\Ten$.
\end{lem}
\begin{proof}
Suppose the joint coloring family for these colorings is $J_C=\{j_c^i:i\}$. We proceed by constructing $\hat{f}$. Assume $(\bx_1^{1},...,\bx_k^{1})\in j_c^i$ and $\cxon(\bx_1^{1})=\sigma_1$, ..., $\cxon(\bx_k^{1})=\sigma_k$. Define $\hat{f}(\sigma_1,...\sigma_k)=f(\bx_1^{1},...,\bx_k^{1})$.

To show this function is well-defined on elements in its support, we should show that for any two points $(\bx_1^{1},...,\bx_k^{1})$ and $(\bx_1^{2},...,\bx_k^{2})$ in $\Ten$, if $\cxon(\bx_1^{1})=\cxon(\bx_1^{2})$, ..., $\cxkn(\bx_k^{1})=\cxkn(\bx_k^{2})$, then $f(\bx_1^{1},...,\bx_k^{1})=f(\bx_1^{2},...,\bx_k^{2})$.  

Since $\cxon(\bx_1^{1})=\cxon(\bx_1^{2})$, ..., $\cxkn(\bx_k^{1})=\cxkn(\bx_k^{2})$, these two points belong to a joint coloring class such as $j_c^i$. Since $\cxon$, ..., $\cxkn$ satisfy C.C.C., by using Lemma \ref{lem:same-value}, $f(\bx_1^{1},...,\bx_k^{1})=f(\bx_1^{2},...,\bx_k^{2})$. Therefore, our function $\hat{f}$ is well-defined and has the desired property.
\end{proof}
Lemma \ref{lem:fhat} implies that, given $\epsilon$-colorings of characteristic graphs of random variables satisfying C.C.C. at the receiver, we can successfully compute the desired function $f$ with a vanishing probability of error as $n$ goes to infinity. Thus, if the decoder at the receiver is given colors, it can look up $f$ based on its table of $\hat{f}$. The question remains of at which rates encoders can transmit these colors to the receiver faithfully (with a probability of error less than $\epsilon$).  

\begin{lem}{(Slepian-Wolf Theorem)}\label{lem:sw}

A rate-region of the network shown in Figure \ref{fig:one-stage} where $f(X_1,...,X_k)=(X_1,...,X_k)$ can be determined by these conditions:  
\begin{equation}\label{eq:doshi}
\forall S\in\cS(k) \Longrightarrow \sum_{i\in S} R_{1i} \geq H (X_{S}|X_{S^{c}}).\\
\end{equation}
\end{lem}

\begin{proof}
See \cite{sw73}.
\end{proof}

We now use the Slepian-Wolf (SW) encoding/decoding scheme on achieved coloring random variables. Suppose the probability of error in each decoder of SW is less than $\frac{\epsilon}{k}$. Then, the total error in the decoding of colorings at the receiver is less than $\epsilon$. Therefore, the total error in the coding scheme of first coloring $\gxon$, ..., $\gxkn$, and then encoding those colors by using SW encoding/decoding scheme is upper bounded by the sum of errors in each stage. By using Lemmas \ref{lem:fhat} and \ref{lem:sw}, the total error is less than $\epsilon$, and goes to zero as $n$ goes to infinity. By applying Lemma \ref{lem:sw} on achieved coloring random variables, we have,  
    
\begin{equation}\label{eq:doshi}
\forall S\in\cS(k) \Longrightarrow \sum_{i\in S} R_{1i} \geq \frac{1}{n} H (\cxsn|\cxscn),
\end{equation}
where $\cxsn$, and $\cxscn$ are $\epsilon$-colorings of characteristic graphs satisfying C.C.C. Thus, using Definition \ref{def:cond} completes the achievability part.

As an example, look at Figure \ref{fig:3part}-c. This network has two source nodes and a receiver. Source nodes compute $\epsilon$-colorings of their characteristic graphs. These colorings should satisfy C.C.C. Then, an SW compression is performed on these colorings. The receiver, first, perform SW decoding to get the colors. Then, by using a look-up table, it can find the value of its desired function (As an example, look at Figure \ref{fig:look-up-ex}).  

\textit{2) Converse}: Here, we show that any distributed functional source coding scheme with a small probability of error induces $\epsilon$-colorings on characteristic graphs of random variables satisfying C.C.C. Suppose $\epsilon > 0$. Define $\cFen$ for all $(n,\epsilon)$ as follows,
\begin{equation}
\cFen=\{\hat{f}:Pr[\hat{f}(\bX_1,...,\bX_k)\neq f(\bX_1,...,\bX_k)]<\epsilon\}.
\end{equation}

In other words, $\cFen$ is the set of all functions equal to $f$ with $\epsilon$ probability of error. For large enough $n$, all achievable functional source codes are in $\cFen$. We call these codes $\epsilon$-achievable functional codes.

\begin{lem}\label{lem:zero-error}
Consider some function $f:\cX_1\times...\times\cX_k\to\cZ$. Any distributed functional code which reconstructs this function with zero error probability induces colorings on $\gxo$,...,$\gxk$ with respect to this function, where these colorings satisfy C.C.C. 
\end{lem}

\begin{proof}
To show this lemma, let us assume we have a zero-error distributed functional code represented by encoders $en_{X_1}$, ..., $en_{X_k}$ and a decoder $r$. Since it is error free, for any two points $(x_1^{1},...,x_k^{1})$ and $(x_1^{2},...,x_k^{2})$, if $p(x_1^{1},...,x_k^{1})>0$, $p(x_1^{2},...,x_k^{2})>0$, $en_{X_1}(x_1^{1})=en_{X_1}(x_1^{2})$, ..., $en_{X_k}(x_k^{1})=en_{X_k}(x_k^{2})$, then,
\begin{equation}\label{eq:r}
f(x_1^{1},...,x_k^{1})=f(x_1^{2},...,x_k^{2})=r'(en_{X_1}(x_1^{1}),...,en_{X_k}(x_k^{1})).
\end{equation}

We want to show that $en_{X_1}$, ..., $en_{X_k}$ are some valid colorings of $\gxo$, ..., $\gxk$ satisfying C.C.C. We demonstrate this argument for $X_1$. The argument for other random variables is analogous. First, we show that $en_{X_1}$ induces a valid coloring on $\gxo$, and then, we show that this coloring satisfies C.C.C. Let us proceed by contradiction. If $en_{X_1}$ did not induce a coloring on $\gxo$, there must be some edge in $\gxo$ with both vertices with the same color. Let us call these vertices $x_1^{1}$ and $x_1^{2}$. Since these vertices are connected in $\gxo$, there must exist a $(x_2^{1},...,x_k^{1})$ such that, $p(x_1^{1},x_2^{1},...,x_k^{1})p(x_1^{2},x_2^{1},...,x_k^{1})>0$,  $en_{X_1}(x_1^{1})=en_{X_1}(x_1^{2})$, and $f(x_1^{1},x_2^{1},...,x_k^{1})\neq f(x_1^{2},x_2^{1},...,x_k^{1})$. By taking $x_2^{1}=x_2^{2}$, ..., $x_k^{1}=x_k^{2}$ in (\ref{eq:r}), one can see that it is not possible. So, the contradiction assumption is wrong and $en_{X_1}$ induces a valid coloring on $\gxo$. 

Now, we should show that these induced colorings satisfy C.C.C. If it was not true, it would mean that there must exist two point $(x_1^{1},...,x_k^{1})$ and $(x_1^{2},...,x_k^{2})$ in a joint coloring class $j_c^i$ such that there is no path between them in $j_c^i$. So, Lemma \ref{lem:equal} says that the function $f$ can get different values in these two points. In other words, it is possible to have $f(x_1^{1},...,x_k^{1})\neq f(x_1^{2},...,x_k^{2})$, where $\cxo(x_1^{1})=\cxo(x_1^{2})$, ..., $\cxk(x_k^{1})=\cxk(x_k^{2})$, which is in contradiction with (\ref{eq:r}). Thus, these colorings satisfy C.C.C.   
\end{proof}
In the last step, we should show that any achievable functional code represented by $\cFen$ induces $\epsilon$-colorings on characteristic graphs satisfying C.C.C.

\begin{lem}\label{lem:sec}
Consider random variables $\bX_1$, ..., $\bX_k$. All $\epsilon$-achievable functional codes of these random variables induce $\epsilon$-colorings on characteristic graphs satisfying C.C.C.
\end{lem}

\begin{proof} 
Suppose $g(\bx_1,...,\bx_k)=r'(en_{X_1}(\bx_1),...,en_{X_k}(\bx_k))\in\cFen$ is such a code. Lemma \ref{lem:zero-error} says that a zero-error reconstruction of $g$ induces some colorings on characteristic graphs satisfying C.C.C., with respect to $g$. Suppose the set of all points $(\bx_1,...,\bx_k)$ such that $g(\bx_1,...,\bx_k)\neq f(\bx_1,...,\bx_k)$ be denoted by $\cC$. Since $g\in\cFen$, $Pr[\cC]<\epsilon$. Therefore, functions $en_{X_1}$, ..., $en_{X_k}$ restricted to $\cC$ are $\epsilon$-colorings of characteristic graphs satisfying C.C.C. (by definition).   
\end{proof}
Lemmas \ref{lem:zero-error} and \ref{lem:sec} establish the converse part and complete the proof.
\end{proof}

If we have two transmitters ($k=2$), Theorem \ref{th:gendoshi} can be simplified as follows.

\begin{cor}\label{th:doshi}
A rate region of the network shown in Figure \ref{fig:3part}-b is determined by these three conditions:
\begin{eqnarray}\label{eq:doshi}
&R_{11}& \geq H_{\gxo} (X_{1}|X_{2})\nonumber\\
&R_{12}& \geq H_{\gxt} (X_{2}|X_{1}) \\
&R_{11}& + R_{12} \geq H_{\sgjt} (X_{1},X_{2}). \nonumber
\end{eqnarray}
\end{cor}

 \begin{figure}[t]
	\centering
    \includegraphics[width=12cm,height=7cm]{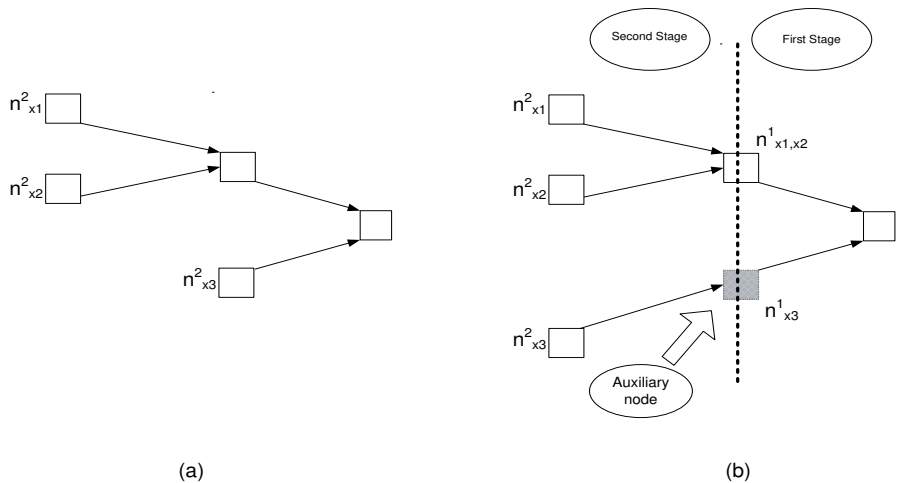}
    \caption{ (a) A simple tree topology. (b) A completed tree.}
    \label{fig:sample-network}
  \end{figure}

\subsection{A Rate Lower Bound for a General Tree Network}\label{sec:tree}

In this section, we compute a rate lower bound of an arbitrary tree network with $k$ correlated sources at its leaves and a receiver in its root (see Figure \ref{fig:general-tree}). We refer to other nodes of this tree as intermediate nodes. The receiver wishes to compute a deterministic function of source random variables. Intermediate nodes have no demand of their own, but they are allowed to perform computation. Computing the desired function $f$ at the receiver is the only demand of the network. For independent sources, we propose a modularized coding scheme to perform arbitrarily closely to derived rate lower bounds.

First, we propose a framework to categorize any tree networks and their nodes.

\begin{defn}
For an arbitrary tree network,
\begin{itemize}
\item The \textit{distance} of each node is the number of hops in the path between that node and the receiver.
\item $d_{max}$ is the distance of the farthest node from the receiver. 
\item A \textit{complete tree} is a tree such that all source nodes are in a distance $d_{max}$ from the receiver.
\item An \textit{auxiliary node} is a new node connected to a leaf of a tree and increases the leaf's distance by one. The added link is called an \textit{auxiliary link}. The leaf in the original tree to which is added an auxiliary node is called the actual node corresponding to that auxiliary node. The link in the original tree connected to the actual node is called the actual link corresponding to that auxiliary link.
\item For any given tree, one can complete it by adding some consecutive auxiliary nodes to its leaves whose distances are less than $d_{max}$. The achieved tree is called the \textit{completed tree} and this process is called \textit{tree completion}.  
\end{itemize}
\end{defn}  

These concepts are depicted in Figure \ref{fig:sample-network}. Auxiliary nodes in the completed tree network act as intermediate nodes. Note that, all functions that may be computed in auxiliary nodes can be gathered in their corresponding actual node in the original tree. So, the rate of the actual link in the original tree network is the minimum of rates of corresponding auxiliary links in the completed tree. Thus, if we compute the rate-region for the completed tree of any given arbitrary tree, we can compute the rate-region of the original tree. Therefore, in the rest of this section, we consider the rate-region of completed tree networks.    

\begin{defn}
Consider a completed tree network with $k$ source nodes in distance $d_{max}$ from the receiver. Nodes in distance $i$ from the receiver are called the $i^{th}$ stage of the tree. $w_i$ is the number of nodes in the $i^{th}$ stage. $\pij$ is a subset of source random variables whose paths to the receiver have the last $i$ common hops. $\pij$ is called source variables of node $n_{ij}$. A connection set of a completed tree is defined as $S_T=\{s_t^i:1 \leq i \leq d_{max}\}$, where $s_t^i=\{\pij:1 \leq j\leq w_i\}$. Note that, any completed tree can be expressed by a connection set.  
\end{defn}

For example, consider the network shown in Figure \ref{fig:sample-network}-b. Its connection set is $S_T=\{s_t^1,s_t^2\}$ such that $s_t^1=\{(X_1,X_2),X_3\}$ and $s_t^2=\{X_1,X_2,X_3\}$. In other words, $\Xi_{11}=(X_1,X_2)$, $\Xi_{12}=X_3$, $\Xi_{21}=X_1$, $\Xi_{22}=X_2$ and $\Xi_{23}=X_3$. One can see that $S_T$ completely describes the structure of the tree.

We have three types of nodes: source nodes, intermediate nodes and a receiver. Source nodes process their messages and transmit them. Intermediate nodes can compute some functions of their received information. The receiver processes the received information to compute its desired function. For example, consider the network shown in Figure \ref{fig:sample-network}-b. Random variables sent through links $e_{21}$, $e_{22}$, $e_{23}$, $e_{11}$ and $e_{12}$ are $M_{21}$, $M_{22}$, $M_{23}$, $M_{11}$ and $M_{12}$ such that $M_{11}=g_{11}(M_{21},M_{22})$, and $M_{12}=g_{12}(M_{23})$.    

\subsubsection{A Rate Lower Bound}
Consider node $n_{ij}$ of a tree. Let $\cS(w_i)$ be the power set of the set $\{1,2,...,w_i\}$ and $s_i\in\cS(w_i)$ be a non-empty subset of $\{1,2,...,w_i\}$.

\begin{thm}\label{thm:tree}
A rate lower bound of a tree network with the connection set $S_T=\{s_t^i:i\}$ can be determined by these conditions,
\begin{equation}\label{eq:rate-one-stage}
\forall s_i\in\cS(w_i) \Longrightarrow \sum_{j\in s_i} R_{ij}\geq H_{\sgxi}(\pis|\pisc)
\end{equation}
for all $i=1,...,|S_T|$ where $\pis=[\pij]_{j\in s_i}$ and $\pisc=\{X_1,...,X_k\}\backslash\{\pis\}$.
\end{thm}
\begin{proof}
In this part, we want to show that no coding scheme can outperform this rate region. Consider nodes in the $i$-th stage of this network, $n_{ij}$ for $1 \leq j\leq w_i$. Suppose they are directly connected to the receiver. So, the information sent in links of this stage should be sufficient to compute the desired function. Suppose their parent nodes sent all their information without doing any compression. So, by direct application of Theorem \ref{th:gendoshi}, (\ref{eq:rate-one-stage}) can be derived.
This argument can be repeated for all stages. Thus, no coding scheme can outperform these bounds.
\end{proof}
In the following, we express some cases under which we can achieve the derived rate lower bound of Theorem \ref{thm:tree}. 
\subsubsection{Tightness of the Rate Lower Bound for Independent Sources}\label{ex:computation}
In this part, we propose a functional coding scheme to achieve the rate lower bound. Suppose random variables $\bX_1$, ..., $\bX_k$ with characteristic graphs $\gxon$, ..., $\gxkn$ are independent. Assume $\cxon$, ..., $\cxkn$ are valid $\epsilon$-colorings of these characteristic graphs satisfying C.C.C. The proposed coding scheme can be described as follows: source nodes first compute colorings of high probability subgraphs of their characteristic graphs satisfying C.C.C., and then, perform source coding on these coloring random variables. Intermediate nodes first compute their parents' coloring random variables, and then, by using a look-up table, find corresponding source values of their received colorings. Then, they compute $\epsilon$-colorings of their own characteristic graphs. The corresponding source values of their received colorings form an independent set in the graph. If all are assigned to a single color in the minimum entropy coloring, intermediate nodes send this coloring random variable followed by a source coding. But, if vertices of this independent set are assigned to different colors, intermediate nodes send the coloring with the lowest entropy followed by source coding (Slepian-Wolf). The receiver first performs a minimum entropy decoding (\cite{kornerbook}) on its received information and achieves coloring random variables. Then, it uses a look-up table to compute its desired function by using achieved colorings.  

To show the achievability, we show that, if nodes of each stage were directly connected to the receiver, the receiver could compute its desired function. Consider the node $n_{ij}$ in the $i$-th stage of the network. Since the corresponding source values of its received colorings form an independent set on its characteristic graph ($G_{\pij}$) and this node computes the minimum entropy of this graph, it is equivalent to the case where it would receive the exact source information, because both of them lead to the same coloring random variable. So, if all nodes of stage $i$ were directly connected to the receiver, the receiver could compute its desired function and link rates would satisfy the following conditions.     

\begin{equation}
\forall s_i\in\cS(w_i) \Longrightarrow\sum_{j\in s_i} R_{ij}\geq H_{\sgxi}(\pis).
\end{equation}

Thus, by using a simple induction argument, one can see that the proposed scheme is achievable and it can perform arbitrarily closely to the derived rate lower bound, while sources are independent. 

\subsection{A Case When Intermediate Nodes Do not Need to Compute}

Though the proposed coding scheme in Section \ref{ex:computation} can perform arbitrarily closely to the rate lower bound, it may require computation at intermediate nodes.
\begin{defn}
Suppose $f(X_1,...,X_k)$ is a deterministic function of random variables $X_1$,...,$X_k$. $(f,X_1,...,X_k)$ is called a chain-rule proper set when for any $s\in\cS(k)$, $H_{\sgs}=H_{G_{X_s}}(X_s)$.
\end{defn}

\begin{thm}
In a general tree network, if sources $X_1$,...,$X_k$ are independent random variables and $(f,X_1,...,X_k)$ is a chain-rule proper set, it is sufficient to have intermediate nodes as relays to perform arbitrarily closely to the rate lower bound mentioned in Theorem \ref{thm:tree}. 
\end{thm}

\begin{proof}
Consider an intermediate node $n_{ij}$ in the $i$-th stage of the network whose corresponding source random variables are $X_s$ where $s\in\cS(k)$ (i.e., $X_s=\pij$). Since random variables are independent, one can write rate bounds of Theorem \ref{thm:tree} as,
\begin{equation}\label{eq:uncond}
\forall s_i\in\cS(w_i) \Longrightarrow\sum_{j\in s_i} R_{ij}\geq H_{\sgxi}(\pis).
\end{equation}

Now, consider the outgoing link rate of the node $n_{ij}$. If this intermediate node acts like a relay, we have $R_{ij}=H_{\sgs(X_s)}$ (since $X_s=\pij$). If $(f,X_1,...,X_k)$ is a chain-rule proper set, we can write, 
\begin{eqnarray}
R_{ij}&=&H_{\sgs}(X_s)\nonumber\\
&=&H_{G_{X_s}}(X_s)\nonumber\\
&=&H_{G_{\pij}}(\pij).
\end{eqnarray}

For any intermediate node $n_{ij}$ where $j\in s_i$ and $s_i\in\cS(w_i)$, we can write a similar argument which lead to conditions (\ref{eq:uncond}).
As mentioned in Theorem \ref{thm:tree}, to perform arbitrarily closely to the rate lower bound, this node needs to compress its received information up to the rate $H_{G_{X_s}}(X_s)$. If this node acted as a relay and forwarded the received information from the previous stage, it would lead to an achievable rate $H_{\bigcup_{i\in s}G_{X_i}}(X_s)$ in the next stage, which in general is not equal to $H_{G_{X_s}}(X_s)$. So, this scheme cannot achieve the rate lower bound. However, if for any $s\in\cS(k)$, $H_{\bigcup_{i\in s}G_{X_i}}(X_s)=H_{G_{X_s}}(X_s)$, this scheme can perform arbitrarily closely to the rate lower bound by having intermediate nodes as relays.   
\end{proof}

In the following Lemma, we provide a sufficient condition to guarantee that a set is a chain-rule proper set.

\begin{lem}\label{thm:proper}
Suppose $X_1$ and $X_2$ are independent and $f(X_1,X_2)$ is a deterministic function. If for any $x_2^1$ and $x_2^2$ in $\cX_2$ we have $f(x_1^i,x_2^1)\neq f(x_1^j,x_2^2)$ for any possible $i$ and $j$, then, $(f,X_1,X_2)$ is a chain-rule proper set. 
\end{lem}

 \begin{figure}[t]
	\centering
    \includegraphics[width=8.5cm,height=5cm]{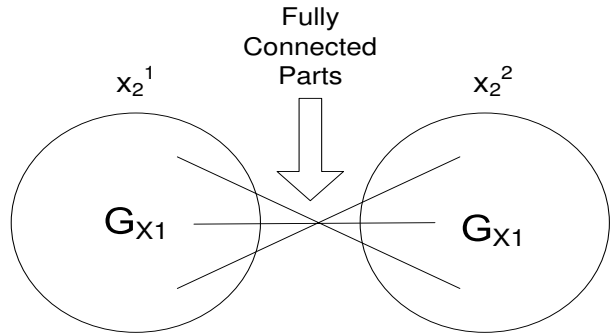}
    \caption{An example of $G_{X_1,X_2}$ satisfying conditions of Lemma \ref{thm:proper}, when $\cX_2$ has two members. }
    \label{fig:color-proper-func}
  \end{figure}

\begin{proof}
We show that under this condition any colorings of the graph $G_{X_1,X_2}$ can be expressed as colorings of $G_{X_1}$ and $G_{X_2}$, and vice versa. The converse part is straightforward because any colorings of $G_{X_1}$ and $G_{X_2}$ can be viewed as a coloring of $G_{X_1,X_2}$.  

Consider Figure \ref{fig:color-proper-func} which illustrates conditions of this lemma. Under these conditions, since all $x_2$ in $\cX_2$ have different function values, graph $G_{X_1,X_2}$ can be decomposed to some subgraphs which have the same topology as $G_{X_1}$, corresponding to each $x_2$ in $\cX_2$. These subgraphs are fully connected to each other under conditions of Corollary \ref{thm:proper}. Thus, any coloring of this graph can be represented as two colorings of $G_{X_1}$ and $G_{X_2}$, which is a complete graph. Hence, the minimum entropy coloring of $G_{X_1,X_2}$ is equal to the minimum entropy coloring of $(G_{X_1},G_{X_2})$. Therefore, $H_{G_{X_1},G_{X_2}}(X_1,X_2)=H_{G_{X_1,X_2}}(X_1,X_2)$.  
\end{proof}

\section{Multi-Functional Compression with Side Information}\label{chap:multi}

In this section, we consider the problem of multi-functional compression with side information. The problem is how we can compress a source $X$ so that the receiver is able to compute some deterministic functions $f_1(X,Y_1)$, ..., $f_m(X,Y_m)$, where $Y_i$, $1\leq i\leq m$, are available at the receiver as side information.

Section \ref{chap:tree} only considers the case where the receiver desires to compute one function (m=1). Here, we consider a case where computation of several functions with different side information is desired at the receiver. Our results do not depend on the fact that all desired functions are in one receiver and one can apply them to the case of having several receivers with different desired functions (i.e., functions are separable). We define a new concept named the \textit{multi-functional graph entropy} which is an extension of the graph entropy defined by K\"orner in \cite{k73}. We show that the minimum achievable rate for this problem is equal to the conditional multi-functional graph entropy of the source random variable given side informations. We also propose a coding scheme based on graph colorings to achieve this rate.

\begin{figure}[t]
	\centering
    \includegraphics[width=9.5cm,height=4cm]{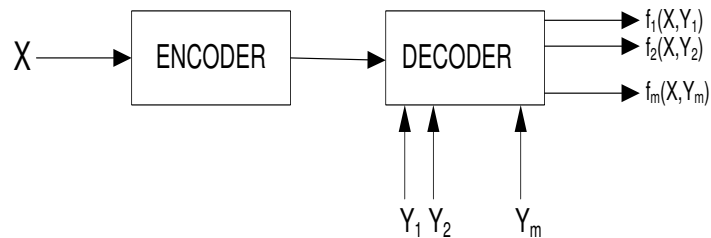}
    \caption{A multi-functional compression problem with side information.}
    \label{fig:sidem}
  \end{figure}

\subsection{Problem Setup} 
Consider discrete memoryless sources (i.e., $\{X^i\}_{i=1}^{\infty}$ and $\{Y_k^i\}_{i=1}^{\infty}$ ) and assume that these sources are drawn from finite sets $\cX$ and $\cY_k$ with a joint distribution $p_k(x,y_k)$. We express $n$-sequence of these random variables as $\mathbf{{X}} = \{X^i\}_{i=l}^{i=l+n-1}$ and $\mathbf{{Y_k}} = \{Y_k\}_{i=l}^{i=l+n-1}$ with joint probability distribution $p_k(\bx,\by_k)$. Assume $k=1,...,m$ (we have $m$ random variables as side information at the receiver). 

The receiver wants to compute $m$ deterministic functions $f_k:\cX\times\cY_k\to\cZ_k$ or $f_k:\cX^{n}\times\cY_{k}^{n}\to\cZ_{k}^{n}$, its vector extension. Without loss of generality, we assume $l=1$ and to simplify notations, $n$ will be implied by the context. We have one encoder $en_X$ and $m$ decoders $r_1,...,r_m$ (one for each function and its corresponding side information). Encoder $en_X$ maps 
\begin{equation}
en_X:\cX^{n}\to\{1,...,2^{nR}\},
\end{equation}
and, each decoder $r_k$ maps
\begin{equation}
r_k:\{1,...,2^{nR}\}\times \{1,...,2^{n}\}\to\cZ_{k}^{n}.
\end{equation}

The probability of error in decoding $f_k$ is
\begin{equation}
P_{e_k}^{n}=Pr[{(\bx,\by_k):f_k(\bx,\by_k)\neq r_k(en_X(\bx),\by_k)}],
\end{equation}
and, the total probability of error is
\begin{equation}
P_{e}^{n}=1-\prod_{k}(1-P_{e_k}^{n}).
\end{equation}
We declare an error when we have an error in computation of at least one function. A rate $R$ is achievable if $P_{e}^{n}\to 0$ when $n\to\infty$. Our aim here is to find the minimum achievable rate.

\subsection{Main Results} \label{sec:Main}

Prior work in the functional compression problem consider a case when computation of a function is desired at the receiver (m=1). In this section, we consider a case when computation of several functions is desired at the receiver. As an example, consider the network shown in Figure \ref{fig:sidem}. The receiver wants to compute $m$ functions with different side information random variables. We want to compute the minimum achievable rate for this case. Note that our results do not depend on the fact that all functions are desired in one receiver and one can extend them to the case of having several receivers with different desired functions (i.e., functions are separable.). 

First, let us consider the case of $m=2$. Then, we extend our results to the case of arbitrary $m$. In this problem, the receiver wants to compute two deterministic functions $f_{1}(X,Y_{1})$ and $f_{2}(X,Y_{2})$, where $Y_{1}$ and $Y_{2}$ are available at the receiver as side information. We wish to find the minimum achievable rate for this problem. 

Let us call $\gf=(V,E_{f_{1}})$ the characteristic graph of $X$ with respect to $Y_{1}$, $p_{1}({x,y_{1}})$ and $f_{1}({X,Y_{1}})$, and $\gff=(V,E_{f_{2}})$ the characteristic graph of $X$ with respect to $Y_{2}$, $p_{2}({x,y_{2}})$ and $f_{2}({X,Y_{2}})$. Now, define $\gm=(V,E_{f_{1},f_{2}})$ such that $E_{f_{1},f_{2}}=E_{f_{1}}\bigcup E_{f_{2}}$. In other words, $\gm$ is the \textit{or} function of $\gf$ and $\gff$. We call $\gm$ the \textit{multi-functional characteristic graph} of $X$.

When we deal with one function, we drop $f$ from notations (as in Section \ref{chap:tree}).

\begin{defn}\label{def:multi-char}
The \textit{multi-functional characteristic graph} $\gm=(V,E_{f_{1},f_{2}})$ of $X$ with respect to $Y_{1}$, $Y_{2}$, $p_{1}(x,y_{1})$, $p_{2}(x,y_{2})$ ,and $f_{1}(x,y_{1})$,$f_{2}(x,y_{2})$ is defined as follows: 

$V = \cX$
and an edge $(x_1,x_2) \in \cX^2$ is in $E_{f_{1},f_{2}}$ iff there exists a $y_{1} \in \cY_{1}$ such that $p_{1}(x_1,y_{1})p_{1}(x_2,y_{1}) > 0$ and $f_{1}(x_1,y_{1})\neq f_{1}(x_2,y_{1})$ or there exists a $y_{2} \in \cY_{2}$ such that $p_{2}(x_1,y_{2})p_{2}(x_2,y_{2}) > 0$ and $f_{2}(x_1,y_{2})\neq f_{2}(x_2,y_{2})$ . 
\end{defn}

Similarly to Definition \ref{def:eq-uncond}, we define the \textit{multi-functional graph entropy} as follows: 

\begin{figure}[t]
	\centering
    \includegraphics[width=10cm,height=5.5cm]{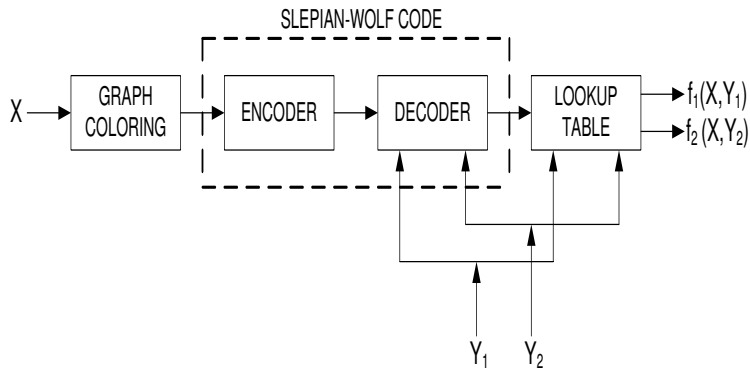}
    \caption{A source coding scheme for multi-functional compression problem with side information.}
    \label{fig:sw}
  \end{figure}

\begin{thm}
\begin{align}
H_{\gm}(X)=\limtoi \frac1n H_{\gm^n}^{\chi}(\bX).
\end{align}
\end{thm}

The \textit{conditional multi-functional graph entropy} can be defined similarly to Definition \ref{def:eq-cond} as follows:

\begin{thm}
\begin{align}
H_{\gm}(X|Y)=\limtoi \frac1n H_{\gm^n}^{\chi}(\bX |\bY),
\end{align}
\end{thm}

where $\gm^n$ is the $n$-th power of $\gm$. Now, we can state the following theorem.

\begin{thm}\label{thm:m2}
$H_{\gm}(X|Y_1,Y_2)$ is the minimum achievable rate for the network shown in Figure \ref{fig:sidem} when $m=2$.
\end{thm}

\begin{proof}
 To show this, we first show that $R_X \geq H_{\gm}(X|Y_1,Y_2)$ is an achievable rate (achievability), and no one can outperform this rate (converse). To do this, first, we show that any valid coloring of $\gm^{n}$ for any $n$ leads to an achievable encoding-decoding scheme for this problem (achievability). Then, we show that every achievable encoding-decoding scheme performing on blocks with length $n$, induces a valid coloring of $\gm^{n}$ (converse). 

Achievablity: According to \cite{doshi-it}, any valid coloring of $\gf^{n}$ leads to successfully computing of $f_{1}(\bX,\bY_{1})$ at the receiver.  If $\cf$ is a valid coloring of $\gf^{n}$, there exists a function $r_{1}$ such that $r_{1}(\cf(\bX),\bY_{1})=f_{1}(\bX,\bY_{1})$, with high probability. A similar argument holds for $\gff$. Now, assume that $\cm$ is a valid coloring of $\gm^{n}$. Since, $E_{f_{1}}^n \subseteq E_{f_{1},f_{2}}^n$ and $E_{f_{2}}^{n} \subseteq E_{f_{1},f_{2}}^{n}$, any valid coloring of $\gm^{n}$ induces valid colorings for $\gf^{n}$ and $\gff^{n}$. Thus, any valid coloring of $\gm^{n}$ leads to successful computation of $f_{1}(\bX,\bY_{1})$ and $f_{2}(\bX,\bY_{2})$ at the receiver. So, $\cm$ leads to an achievable encoding scheme (i.e. there exist two functions $r_{1}$ and $r_{2}$ such that $r_{1}(\cm(\bX),\bY_{1})=f_{1}(\bX,\bY_{1})$ and $r_{2}(\cm(\bX),\bY_{2})=f_{2}(\bX,\bY_{2})$, with high probability.). 

When the receiver wants the whole information of the source node, Slepian and Wolf proposed a technique in \cite{sw73} to compress source random variable $X$ up to the rate $H(X|Y)$ when $Y$ is available at the receiver. Here, one can perform Slepian-Wolf compression technique on the minimum entropy coloring of large enough power graph and get the given bound.  

Converse: Now, we show that any achievable encoding-decoding scheme performing on blocks with length $n$, induces a valid coloring of $\gm^{n}$. In other words, we want to show that if there exist functions $en_X$, $r_{1}$ and $r_{2}$ such that $r_{1}(en_X(\bX),\bY_{1})=f_{1}(\bX,\bY_{1})$ and $r_{2}(en_X(\bX),\bY_{2})=f_{2}(\bX,\bY_{2})$, $en_X(\bX)$ is a valid coloring of $\gm^{n}$.  

Let us proceed by contradiction. If $en_X(\bX)$ were not a valid coloring of $\gm^{n}$, there must be some edge in $E_{f_{1},f_{2}}^{n}$ with both vertices with the same color. Let us call these two vertices $\bx_1$ and $\bx_2$ which take the same values  (i.e., $en_X(\bx_1)=en_X(\bx_2)$), but also are connected. Since they are connected to each other, by definition of $\gm^{n}$, there exists a $\by_1\in\cY_{1}$ such that $p_{1}(\bx_1,\by_{1})p_{1}(\bx_2,\by_{1}) > 0$ and $f_{1}(\bx_1,\by_{1})\neq f_{1}(\bx_2,\by_{1})$, or there exists a $\by_2\in\cY_{2}$ such that $p_{2}(\bx_1,\by_{2})p_{2}(\bx_2,\by_{2}) > 0$ and $f_{2}(\bx_1,\by_{2})\neq f_{1}(\bx_2,\by_{2})$. Without loss of generality, assume that the first case occurs. Thus, we have a $\by_{1}\in \cY_{1}$ such that $p_{1}(\bx_1,\by_{1})p_{1}(\bx_2,\by_{1}) > 0$ and $f_{1}(\bx_1,\by_{1})\neq f_{1}(\bx_2,\by_{1})$. So, $r_{1}(en_X(\bx_1),\by_{1})\neq r_{1}(en_X(\bx_2),\by_{1})$. Since  $en_X(\bx_1)=en_X(\bx_2)$, then, $r_{1}(en_X(\bx_1),\by_{1})\neq r_{1}(en_X(\bx_1),\by_{1})$. But, it is not possible. Thus, our contradiction assumption was not true. In other words, any achievable encoding-decoding scheme for this problem induces a valid coloring of $\gm^{n}$ and thus completes the proof.
\end{proof}

Now, let us consider the network shown in Figure \ref{fig:sidem}, where the receiver wishes to compute $m$ deterministic functions of source information having some side information.  

\begin{thm}
$H_{\gmm}(X|Y_1,...,Y_m)$ is the minimum achievable rate for the network shown in Figure \ref{fig:sidem}.
\end{thm}

The argument here is similar to the case of $m=2$ given in Theorem \ref{thm:m2}. We only sketch the proof. One may first show that any colorings of multi-functional characteristic graph of $X$ with respect to desired functions (i.e., $\gmm$) leads to an achievable scheme. Then, showing that any achievable encoding-decoding scheme induces a coloring on $\gmm$ completes the proof. 


\section{Feedback in Functional Compression}\label{chap:feedback}

In this section, we investigate the effect of having feedback on the rate-region of the functional compression problem. If the function at the receiver is the identity function, this problem is Slepian-Wolf compression with feedback. For this case, having feedback does not improve rate bounds. For example, reference \cite{mayak} considers both zero-error and asymptotically zero-error Slepian-Wolf compression with feedback. However, for a general desired function at the receiver, having feedback may improve rate bounds of the case without feedback.

\subsection{Main Results}
Consider a distributed functional compression problem with two sources and a receiver depicted in Figure \ref{fig:net-feedback}-a. This network does not have feedback. In Section \ref{chap:tree}, we derived a rate-region for this network. In this section, we consider the effect of having feedback on the rate-region of the network. For simplicity, we consider a simple distributed network topology with two sources. However, one can extend all discussions to more general networks of the type considered in Sections \ref{chap:tree} and \ref{chap:multi}.

Consider the network shown in Figure \ref{fig:net-feedback}-b. If the desired function at the receiver is the identity function, this problem is Slepian-Wolf compression with feedback. For this case, having feedback does not change the rate region (\cite{sw73} and \cite{mayak}). 
However, when we have a general function at the receiver, by having feedback, one may improve the rate bounds of Theorem \ref{th:doshi}.

 \begin{figure}[t]
	\centering
    \includegraphics[width=9cm,height=5cm]{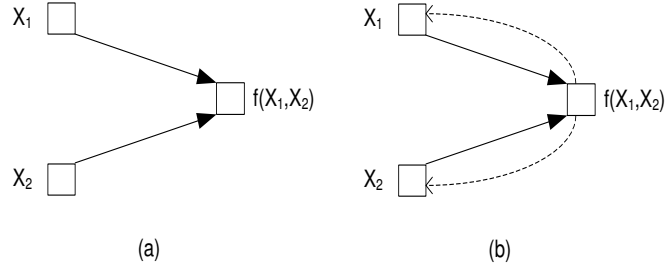}
    \caption{A distributed functional compression network a) without feedback b) with feedback.}
    \label{fig:net-feedback}
  \end{figure}

\begin{thm}
Having feedback may improve rate bounds of Theorem \ref{th:doshi}.  
\end{thm}

\begin{proof}
Consider a network without feedback depicted in Figure \ref{fig:net-feedback}-a. In Section \ref{chap:tree}, we showed an achievable scheme where sources send their minimum entropy colorings of high probability subgraphs of their characteristic graphs satisfying C.C.C., followed by Slepian-Wolf compression. This scheme performs arbitrarily closely to rate bounds derived in Theorem \ref{th:doshi}. Now, we seek to show that, in some cases, by having feedback, one can outperform these bounds. Consider source random variables $X_1$ and $X_2$ with characteristic graphs $\gxo$ and $\gxt$, respectively. Suppose $S_{\cjmin}$ and $S_{\cjsmin}$ are two sets of joint colorings of source random variables defined as follows,

\begin{eqnarray}
S_{\cjmin}= \qquad\arg\!\!\!\!\!\!\!\!\!\!\!\!\!\!\!\!\!\!\!\!\!\min_{\substack{(\cxon,\cxtn)\in\Cxon\times\Cxtn}} \frac{1}{n} H(\cxon,\cxtn)\nonumber\\
S_{\cjsmin}=\qquad\arg\!\!\!\!\!\!\!\!\!\!\!\!\!\!\!\!\!\!\!\!\!\min_{\substack{(\cxon,\cxtn)\in\Cxon\times\Cxtn\\
\textrm{ satisfying C.C.C.}}} \frac{1}{n} H(\cxon,\cxtn).
\end{eqnarray}

Now, consider the case when $S_{\cjmin}\cap S_{\cjsmin}=\emptyset$, i.e., suppose any $\cjmin\in S_{\cjmin}$ does not satisfy C.C.C. Thus, C.C.C. restricts the link sum rates of any achievable scheme, because $H(\cjmin)<H(\cjsmin)$ for any $\cjmin\in S_{\cjmin}$ and $\cjsmin\in S_{\cjsmin}$.

Choose any two joint colorings $\cjmin\in S_{\cjmin}$ and $\cjsmin\in S_{\cjsmin}$. Suppose set $A$ contains all points $(\bx_1,\bx_2)$ such that their corresponding colors in the joint-coloring class of $\cjmin$ do not satisfy C.C.C. Now, we propose a coding scheme with feedback which can outperform rate bounds of the case without having feedback. If sources know whether or not they have some sequences in $A$, they can switch between $\cjmin$ and $\cjsmin$ in their coding scheme with feedback. Since $H(\cjmin)<H(\cjsmin)$, this approach outperforms the one without feedback in terms of rates. In the following, we present a possible feedback scheme.

Before sending each sequence, sources first check if their sequences belong to $A$ or not. Say $A_{X_1}$ is the set of all $\bx_1$ such that there exists a $\bx_2$ such that $(\bx_1,\bx_2)\in A$. $A_{X_2}$ is defined similarly. One can see that $A\subseteq A_{X_1}\times A_{X_2}$. So, instead of checking if a sequence is in $A$ or not, by exchanging some information, sources check if the sequence belongs to $A_{X_1}\times A_{X_2}$ or not. In order to do this, source $X_1$ sends a one to the receiver when $\bx_1\in A_{X_1}$. Otherwise, it sends a zero. Source $X_2$ uses a similar scheme. The receiver exchanges these bits using feedback channels. When a source sends a one, and receives a one from its feedback channel, it uses $\cjsmin$ as its joint coloring. Otherwise, it uses $\cjmin$ in its coding scheme. Depending on which joint coloring scheme has been used by sources, the receiver uses a corresponding look-up table to compute the desired function. Hence, this scheme is achievable. An example of this scheme is depicted in Figure \ref{fig:scheme-feedback}.

\begin{figure}[t]
	\centering
    \includegraphics[width=9cm,height=5cm]{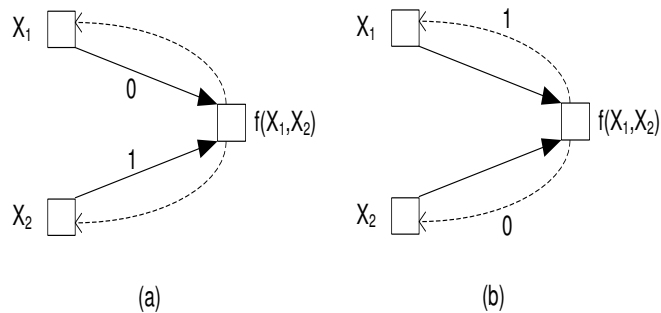}
    \caption{An example of the proposed feedback scheme. a) Since $\bx_1\notin A_{X_1}$, source $X_1$ sends 0. Since $\bx_2\in A_{X_2}$, source $X_2$ sends 1. b) The receiver forward signaling bits to the sources. Then, sources can use the coloring scheme $\cjmin$.}
    \label{fig:scheme-feedback}
  \end{figure}

Since the length of sequences is arbitrarily large, one can ignore these four extra signaling bits in rate computation. If we did not have feedback, according to Theorem \ref{th:doshi},
\begin{equation}\label{eq:without-feedback}
R_{11}+R_{12}\geq \frac{1}{n} H(\cjsmin).
\end{equation}     

Say $P_a=Pr[(\bx_1,\bx_2)\in A_{X_1}\times A_{X_2}]$. Thus, for the proposed coding scheme with feedback, we have,
\begin{equation}\label{eq:with-feedback}
R_{11}^{f}+R_{12}^{f}\geq \frac{1}{n}[P_a H(\cjsmin)+ (1-P_a) H(\cjmin)]
\end{equation}
where $R_{1i}^f$ is the transmission rate of source $i$ with feedback. Thus,
\begin{equation}\label{eq:gain}
[R_{11}^{f}+R_{12}^{f}]-[R_{11}+R_{12}]\geq \frac{1}{n} (1-P_a) [H(\cjmin)- H(\cjsmin)].
\end{equation}

The right-hand side of (\ref{eq:gain}) represents a gain in link sum rates by having feedback. When, $P_a\neq 1$ and $\cjmin\neq\cjsmin$, this is strictly positive, which means the proposed coding scheme with feedback outperforms the one without having feedback in terms of rate bounds. For the identity function at the receiver, $\cjmin=\cjsmin$, and the proposed coding scheme with feedback does not improve rate bounds. Note that, for the identity function at the receiver, Slepian-Wolf compression can perform arbitrarily closely to min-cut max-flow bounds.
\end{proof}

Note that, in a general network, for cases where the minimum entropy colorings of sources satisfy C.C.C., it is not known whether or not feedback can improve rate bounds.

\section{A Rate-Distortion Region for Distributed Functional Compression}\label{chap:distortion}

In this section, we consider the problem of distributed functional compression with distortion. The objective is to compress correlated discrete sources so that an arbitrary deterministic function of those sources can be computed up to a distortion level at the receiver. In this section, we derive a rate-distortion region for a network with two transmitters and a receiver. All discussions can be extended to more general networks considered in Sections \ref{chap:tree} and \ref{chap:multi}. 

A recent result is presented in \cite{doshi-it} which computes a rate-distortion region for the side information problem. The result in \cite{doshi-it} gives a characterization of Yamamoto's rate distortion function \cite{yam} in terms of a reconstruction function. Here, we extend these results to the distributed functional compression problem. In this case, we compute a rate-distortion region and then, propose a practical coding scheme with a non-trivial performance guarantee. Note that this proposed characterization is not a single letter characterization.

\subsection{Problem Setup} \label{subsec:problemsetup}
Consider two sources as described in Section \ref{subsec:problemsetup}. The receiver wants to compute a deterministic function $f:\cX_1\times\cX_2\to\cZ$ or $f:\cX_1^{n}\times\cX_2^{n}\to\cZ^{n}$, its vector extension up to distortion $D$ with respect to a given distortion function $d: \cZ\times\cZ\to[0,\infty)$. A vector extension of the distortion function is defined as follows:

\begin{equation}
d(\bz_{1},\bz_{2})=\frac{1}{n}\sum_{i=1}^{n} d(z_{1i},z_{2i}),
\end{equation}

where $\bz_{1},\bz_{2}\in\cZ^{n}$. As in \cite{wz76}, we assume that $d(z_{1},z_{2})=0$ if and only if $z_{1}=z_{2}$. This assumption causes vector extension to satisfy the same property (i.e., $d(\bz_{1},\bz_{2})=0$ if and only if $\bz_{1}=\bz_{2}$). 

Consider the network depicted in Figure \ref{fig:3part}-b. The sources encode their data at rates $R_{11}$ and $R_{12}$ by using encoders $en_{X_1}$ and $en_{X_2}$, respectively . The receiver decodes the received data by using decoder $r$. Hence, we have:

\begin{eqnarray}
en_{X_1}:\cX_1^{n}\to\{1,...,2^{nR_{11}}\} \nonumber\\
en_{X_2}:\cX_2^{n}\to\{1,...,2^{nR_{12}}\} \nonumber
\end{eqnarray}
 and a decoder maps,
\begin{equation}
r:\{1,...,2^{nR_{11}}\}\times\{1,...,2^{nR_{12}}\}\to\cZ^{n}. \nonumber
\end{equation}

The probability of error is
\begin{align*}
P_{e}^{n} = \Pr[\{(\bx_1,\bx_2):d(f(\bx_1,\bx_2),r(en_{X_1}(\bx_1),en_{X_2}(\bx_2)))> D\}].
\end{align*}

We say a rate pair $(R_{11},R_{12})$ is achievable up to distortion $D$ if there exist $en_{X_1}$, $en_{X_2}$ and $r$ such that $P_{e}^{n}\to 0$ when $n\to\infty$.

Our aim is to find feasible rates for different links of the network shown in Figure \ref{fig:3part}-b when the receiver wants to compute $f(X_1,X_2)$ up to distortion $D$.

\begin{table}
\begin{center}
	\caption{Research progress on nonzero-distortion source coding problems}\label{tab}
  \begin{tabular}{ccc}
\toprule
\textbf{Problem types} & $f(X_1,X_2)=(X_1,X_2)$ & General $f(X_1,X_2)$ \\
\toprule
\multirow{3}{*}{\textbf{Side information}} &
 & Feng et al. \cite{fes04}\\ 
&Wyner and Ziv \cite{wz76} & Yamamoto \cite{yam} \\ 
& & Doshi et al. \cite{doshi-it} \\
\midrule
\multirow{4}{*}{\textbf{Distributed}} & 
Coleman et al. \cite{coleman} & \\
&Berger and Yeung \cite{by89} & * \\
& Barros and Servetto \cite{bs03} & \\
& Wagner et al. \cite{wtv06} & \\
\bottomrule
\end{tabular}
\end{center}
\end{table}

\subsection{Prior Results} \label{subsec:previous}
In this part, we overview prior relevant work. Consider the network shown in Figure \ref{fig:3part}-a. For this network, in \cite{yam}, Yamamoto gives a characterization of a rate-distortion function for the side information functional compression problem (i.e., $X_2$ is available at the receiver). The rate-distortion function proposed in \cite{yam} is a generalization of the Wyner-Ziv side-information rate-distortion function \cite{wz76}. Specifically, Yamamoto gives the rate distortion function as follows:

\begin{thm}\label{thm:yam}
The rate distortion function for the functional compression problem with side information is
\[R(D)=\min_{p\in\cP(D)}I(W_1;X_1|X_2)\]
where $\cP(D)$ is the collection of all distributions on $W_1$ given $X_1$ such that
there exists a $g:\cW_1\times\cX_2\to\cZ$ satisfying $E[d(f(\bX_1,\bX_2),g(\bW_1,\bX_2))]\leq D$.
\end{thm}

This is an extension of the Wyner-Ziv rate-distortion result \cite{wz76}. Further, the variable $W_1\in\Gamma(G_{X_1})$ in the definition of the Orlitsky-Roche rate, Definition \ref{def:or}, (a variable over the independent sets of $G_{X_1}$) can be seen as an  interpretation of Yamamoto's auxiliary variable, $W_1$, for the zero-distortion case.

A new characterization of the rate distortion function given by Yamamoto was discussed in \cite{doshi-it}. It was shown in \cite{doshi-it} that finding a suitable reconstruction function, $\fh$, is equivalent
to find $g$ on $\cW_1\times\cX_2$ from Theorem \ref{thm:yam}.
Let $\cF_m(D)$ denote the set of all functions $\fh_m:\cX_1^m\times\cX_2^m\to\cZ^m$ such that 
\[\lim_{n\to\infty} E[d(f(\bX_1,\bX_2),\fh_m(\bX_1,\bX_2))]\leq D, \]
and let $\cF(D)=\bigcup_{m\in {N}}\cF_m(D)$. Also, let $\gxofh$ denote the characteristic graph of $\bX_1$ with respect to $\bX_2$, $p(\bx_1,\bx_2)$, and $\fh$ for any $\fh\in\cF(D)$. For each $m$ and all functions $\fh\in\cF(D)$, denote
for brevity the normalized graph entropy $\frac{1}{m} H_{\gxofh}(\bX_1|\bX_2)$
as $H_{\gxofh}(X_1|X_2)$. The following theorem was given in \cite{doshi-it}:

\begin{thm}\label{thm:doshi-dist}
A rate distortion function for the network shown in Figure \ref{fig:3part}-a can be expressed as follows:
\[R(D)=\inf_{\fh\in\cF(D)}H_{\gxofh}(X_1|X_2).\]
\end{thm}

The problem of finding an appropriate function $\fh$ is equivalent to finding a new graph whose edges are a subset of the edges of the characteristic graph. A graph parameterization by $D$ was proposed in \cite{doshi-it} to look at a subset of $\cF(D)$. The resulting bound is not tight, but it provides a practical technique to tackle a very difficult problem.

Define the $D$-characteristic graph of $X_1$ with respect to $X_2$, $p(x_1,x_2)$, and $f(X_1,X_2)$, as having vertices $V=\cX_1$ and the pair $(x_1^1,x_1^2)$ is an edge if there exists some $x_2^1\in\cX_2$ such that $p(x_1^1,x_2^1)p(x_1^2,x_2^1)>0$ and $d(f(x_1^1,x_2^1),f(x_1^2,x_2^1))>D$.  We call this graph as $G_{X_1}(D)$. Because $d(z_1,z_2)=0$ if and only if $z_1=z_2$, the $0$-characteristic graph is the characteristic graph (i.e., $G_{X_1}(0)=G_{X_1}$). The following corollary was given in \cite{doshi-it}:

\begin{cor}\label{thm:cor}
The rate $H_{G_{X_1}(D)}(X_1|X_2)$ is achievable.
\end{cor}

\subsection{Main Results} \label{sec:main}
This section contains our contributions in this problem. Our aim is to find a rate-distortion region for the network shown in Figure \ref{fig:3part}-b. Recall the Yamamoto rate distortion function (Theorem \ref{thm:yam}) and Theorem \ref{thm:doshi-dist}. These theorems explain a rate distortion function for the side information problem. Now, we are considering the case when we have distributed functional compression. 

Again, for any $m$, let $\cF_m(D)$ denote the set of all functions $\fh_m:\cX_1^m\times\cX_2^m\to\cZ^m$ such that 
\[\lim_{n\to\infty} E[d(f(\bX_1,\bX_2),\fh_m(\bX_1,\bX_2))]\leq D. \]

In other words, we consider $n$ blocks of $m$-vectors; thus, the functions in the expectation above will be on $\cX_1^{mn}\times\cX_2^{mn}$. Let $\cF(D)=\bigcup_{m\in{N}}\cF_m(D)$. Let $\gxofh$ denote the characteristic graph of $\bX_1$ with respect to $\bX_2$, $p(\bx_1,\bx_2)$, and $\fh$ for any $\fh\in\cF(D)$ and $\gxtfh$ denote the characteristic graph of $\bX_2$ with respect to $\bX_1$, $p(\bx_1,\bx_2)$, and $\fh$ for any $\fh\in\cF(D)$. For each $m$ and all functions $\fh\in\cF(D)$, denote
for brevity normalized graph entropies $\frac{1}{m} H_{\gxofh}(\bX_1|\bX_2)$
as $H_{\gxofh}(X_1|X_2)$, $\frac{1}{m} H_{\gxtfh}(\bX_2|\bX_1)$
as $H_{\gxtfh}(X_2|X_1)$ and $\frac{1}{m} H_{\gxofh,\gxtfh}(\bX_1,\bX_2)$
as $H_{\gxofh,\gxtfh}(X_1,X_2)$.  

Now, for a specific function $\fh\in\cF(D)$, define $R_{\fh}(D)=(R_{11}^{\fh}(D),R_{12}^{\fh}(D))$ such that,
\begin{eqnarray}\label{eq:cond}
R_{11}^{\fh}&\geq& H_{\gxofh}(X_1|X_2) \\
R_{12}^{\fh}&\geq& H_{\gxtfh}(X_2|X_1) \nonumber\\
R_{11}^{\fh}+R_{12}^{\fh}&\geq& H_{\gxofh,\gxtfh}(X_1,X_2). \nonumber
\end{eqnarray}

\begin{thm}\label{thm:main}
A rate-distortion region for the network shown in Figure \ref{fig:3part}-b is determined by $\bigcup_{\fh\in\cF(D)} R_{\fh}(D)$. 
\end{thm}

\begin{proof}
We want to show that $\bigcup_{\fh\in\cF(D)} R_{\fh}(D)$ determines a rate-distortion region for the considered network. We first show this rate-distortion region is achievable for any $\fh\in\cF(D)$, and then we prove every achievable rate region is a subregion of it (converse).   

According to Theorem \ref{th:doshi}, $R_{\fh}(D)$ is sufficient to determine the function $\fh(\bX_1,\bX_2)$ at the receiver. Also, by definition,
\[\lim_{n\to\infty} E[d(f(\bX_1,\bX_2),\fh(\bX_1,\bX_2))]\leq D. \]

Thus, for a specific $\fh\in\cF(D)$, $R_{\fh}(D)$ is achievable. Therefore, the union of these achievable regions for different $\fh\in\cF(D)$ \big(i.e., $\bigcup_{\fh\in\cF(D)} R_{\fh}(D)$\big) is also achievable.

Next, we show that any achievable rate region is a subregion of $\bigcup_{\fh\in\cF(D)} R_{\fh}(D)$. Assume that we have an achievable scheme in which source $1$ encodes its data to $en_{X_1}(\bX_1)$ and source $2$ encodes its data to $en_{X_2}(\bX_2)$. At the receiver, we compute $r(en_{X_1}(\bX_1),en_{X_2}(\bX_2))$. Since it is an achievable scheme up to distortion $D$, there exists $\fh\in\cF(D)$ such that $r(en_{X_1}(\bX_1),en_{X_2}(\bX_2))=\fh(\bX_1,\bX_2)$. Thus, considering Theorem \ref{th:doshi}, this achievable rate-distortion region is a subregion of $\bigcup_{\fh\in\cF(D)} R_{\fh}(D)$. This completes the proof.
\end{proof}

Next, we present a simple scheme which satisfies Theorem \ref{thm:main}. Again, the problem of finding an appropriate function $\fh$ is equivalent to finding a new graph whose edges are a subset of the edges of the characteristic graph of random variables. This motivates Corollary \ref{thm:mycor} where we use a similar graph parameterization by $D$. Our scheme is as follows:

Define the $D$-characteristic graph of $X_1$ with respect to $X_2$, $p(x_1,x_2)$, and $f(X_1,X_2)$, as having vertices $V=\cX_1$ and the pair $(x_1^1,x_1^2)$ is an edge if there exists some $x_2^1\in\cX_2$ such that $p(x_1^1,x_2^1)p(x_1^2,x_2^1)>0$ and $d(f(x_1^1,x_2^1),f(x_1^2,x_2^1))>D$.  Denote this graph as $G_{X_1}(D)$. Similarly, we define $G_{X_2}(D)$. Following Corollary \ref{thm:cor} and Theorem \ref{thm:main}, we have the following Corollary:

\begin{cor}\label{thm:mycor}
For independent sources, if the distortion function is a metric and $(R_{11},R_{12})$ satisfies the following conditions, then, $(R_{11},R_{12})$ is achievable.

\begin{eqnarray}\label{eq:cond2}
R_{11}&\geq& H_{G_{X_1}(D/2)}(X_1) \\
R_{12}&\geq& H_{G_{X_2}(D/2)}(X_2) \nonumber\\
R_{11}+R_{12}&\geq& H_{G_{X_1}(D/2),G_{X_2}(D/2)}(X_1,X_2). \nonumber
\end{eqnarray} 
\end{cor}

\begin{proof}
From Theorem \ref{th:doshi}, by sending colorings of high probability subgraphs of sources's $D/2$-characteristic graphs satisfying C.C.C., one can achieve the rate region described in (\ref{eq:cond2}). For simplicity, we assume the power of the graphs is one. Extensions to an arbitrary power are analogous. Suppose the receiver gets two colors from sources (say $c_1$ from source $1$, and $c_2$ from source $2$). To show that the receiver is able to compute its desired function up to distortion level $D$, we need to show that for every $(x_1^1,x_2^1)$ and $(x_1^2,x_2^2)$ such that $C_{G_{X_1}(D/2)}(x_1^1)=C_{G_{X_1}(D/2)}(x_1^2)$ and $C_{G_{X_2}(D/2)}(x_2^1)=C_{G_{X_2}(D/2)}(x_2^2)$, we have $d(f(x_1^1,x_2^1),f(x_1^2,x_2^2))\leq D$. Since the distortion function $d$ is a metric, we have,
\begin{eqnarray}
d(f(x_1^1,x_2^1),f(x_1^2,x_2^2))&\leq& d(f(x_1^1,x_2^1),f(x_1^2,x_2^1))+ d(f(x_1^2,x_2^1),f(x_1^2,x_2^2))\nonumber\\
&\leq& D/2+D/2=D. 
\end{eqnarray}
This completes the proof.
\end{proof}

\section{Polynomial Time Cases for Finding the Minimum Entropy Coloring of a Characteristic Graph}\label{chap:min-coloring}

In this section, we consider the problem of finding the minimum entropy coloring of a characteristic graph. This problem arises in the functional compression problem where computation of a function of sources is desired at the receiver. We considered some aspects of this problem in Sections [\ref{chap:tree}- \ref{chap:distortion}] and proposed some coding schemes. In those proposed coding schemes, one needs to compute the minimum entropy coloring (a coloring random variable which minimizes the entropy) of a characteristic graph. In general, finding this coloring is an NP-hard problem (as shown by Cardinal et al. \cite{np1}) . However, in this section, we show that depending on the characteristic graph's structure, there are some interesting cases where finding the minimum entropy coloring is not NP-hard, but tractable and practical. In one of these cases, we show that, having a non-zero joint probability condition on random variables' distributions, for any desired function $f$, makes characteristic graphs to be formed of some non-overlapping fully-connected maximal independent sets. Therefore, the minimum entropy coloring can be solved in polynomial time. In another case, we show that, if the function we seek to compute is a type of quantization functions, this problem is also tractable. 

We also consider this problem in a general case. By using Huffman or Lempel-Ziv coding notions, we heuristically relate finding the minimum entropy coloring to finding the maximum independent set of a graph. While the minimum-entropy coloring problem is a recently studied problem, there are some heuristic algorithms to solve approximately the maximum independent set problem.

We proceed this section by stating the problem setup. Then, we explain our contributions to this problem.

\subsection{Problem Setup}

In some problems such as the functional compression problem, we need to find a coloring random variable of a characteristic graph which minimizes the entropy. The problem is how to compute such a coloring for a given characteristic graph. 

Given a characteristic graph $\gxo$ (or, its $n$-th power, $\gxon$), one can assign different colors to its vertices. Suppose $\Cxo$ is the collection of all valid colorings of this graph. Among these colorings, one which minimizes the entropy of the coloring random variable is called the minimum-entropy coloring, and we refer to it by $\cxomin$:
\begin{equation}
\cxomin=\arg\!\!\!\!\!\!\!\!\min_{\cxo\in\Cxo} H(\cxo).
\end{equation}

The problem is how to compute $\cxomin$ given $\gxo$. 
 
\subsection{Main Results}\label{sec:coloring}

In general, finding $\cxomin$ is an NP-hard problem (\cite{np1}). However, in this section, we investigate cases where this coloring can be computed in polynomial time, depending on the characteristic graph's structure. In one of these cases, we show that, by having a non-zero joint probability condition on random variables' distributions, for any desired function, finding $\cxomin$ can be solved in polynomial time. In another case, we show that, if the function we seek to compute is a quantization function, this problem is also tractable. 
We also consider this problem in a general case. By using Huffman or Lempel-Ziv coding notions, we heuristically relate finding the minimum entropy coloring to finding the maximum independent set of a graph.

For simplicity, we consider functions with two input random variables, but one can extend all discussions to functions with more input random variables than two.

  \begin{figure}[t]
	\centering
    \includegraphics[width=10cm,height=5cm]{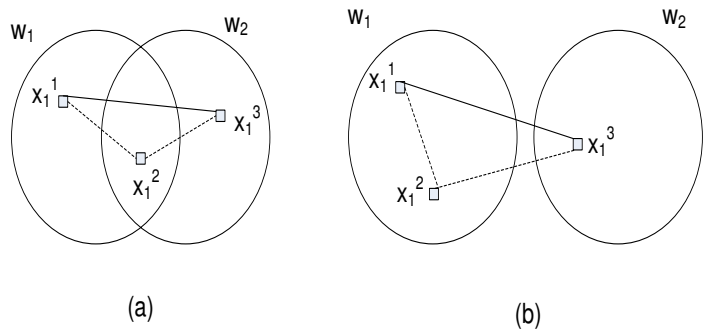}
    \caption{ Having non-zero joint probability distribution, a) maximal independent sets cannot overlap with each other (this figure is to depict the contradiction) b) maximal independent sets should be fully connected to each other. In this figure, a solid line represents a connection, and a dashed line means no connection exists.}
    \label{fig:indep-sets}
  \end{figure}

\subsubsection{Non-Zero Joint Probability Distribution Condition} \label{subsec:trivial-coloring}  

Consider the network shown in Figure \ref{fig:3part}-b. Source random variables have a joint probability distribution $p(x_1,x_2)$, and the receiver wishes to compute a deterministic function of sources (i.e., $f(X_1,X_2)$). In Section \ref{chap:tree}, we showed that, in an achievable coding scheme, one needs to compute minimum entropy colorings of characteristic graphs. The question is how source nodes can compute minimum entropy colorings of their characteristic graphs $\gxo$ and $\gxt$ (or, similarly the minimum entropy colorings of $\gxon$ and $\gxtn$, for some $n$). For an arbitrary graph, this problem is NP-hard (\cite{np1}). However, in certain cases, depending on the probability distribution or the desired function, the characteristic graph has some special structure which leads to a tractable scheme to find the minimum entropy coloring. In this section, we consider the effect of the probability distribution. 

\begin{thm}\label{thm:color}
Suppose for all $(x_1,x_2)\in\cX_1\times\cX_2$, $p(x_1,x_2)>0$. Then, maximal independent sets of the characteristic graph $\gxo$ (and, its $n$-th power $\gxon$, for any $n$) are some non-overlapping fully-connected sets. Under this condition, the minimum entropy coloring can be achieved by assigning different colors to its different maximal independent sets. 
\end{thm}

\begin{proof}
Suppose $\Gamma(\gxo)$ is the set of all maximal independent sets of $\gxo$. Let us proceed by contradiction. Consider Figure \ref{fig:indep-sets}-a. Suppose $w_1$ and $w_2$ are two different non-empty maximal independent sets. Without loss of generality, assume $x_1^{1}$ and $x_1^2$ are in $w_1$, and $x_1^{2}$ and $x_1^3$ are in $w_2$. These sets have a common element $x_1^2$. Since $w_1$ and $w_2$ are two different maximal independent sets, $x_1^{1}\notin w_2$ and $x_1^{3}\notin w_1$. Since $x_1^1$ and $x_1^2$ are in $w_1$, there is no edge between them in $\gxo$. The same argument holds for $x_1^2$ and $x_1^3$. But, we have an edge between $x_1^1$ and $x_1^3$, because $w_1$ and $w_2$ are two different maximal independent sets, and at least there should exist such an edge between them. Now, we want to show that it is not possible.

Since there is no edge between $x_1^1$ and $x_1^2$, for any $x_2^1\in\cX_2$, $p(x_1^1,x_2^1)p(x_1^2,x_2^1)>0$, and $f(x_1^1,x_2^1)=f(x_1^2,x_2^1)$. A similar argument can be expressed for $x_1^2$ and $x_1^3$. In other words, for any $x_2^1\in\cX_2$, $p(x_1^2,x_2^1)p(x_1^3,x_2^1)>0$, and $f(x_1^2,x_2^1)=f(x_1^3,x_2^1)$. Thus, for all $x_2^1\in\cX_2$, $p(x_1^1,x_2^1)p(x_1^3,x_2^1)>0$, and $f(x_1^1,x_2^1)=f(x_1^3,x_2^1)$. However, since $x_1^1$ and $x_1^3$ are connected to each other, there should exist a $x_2^1\in\cX_2$ such that $f(x_1^1,x_2^1)\neq f(x_1^3,x_2^1)$ which is not possible. So, the contradiction assumption is not correct and these two maximal independent sets do not overlap with each other.

We showed that maximal independent sets cannot have overlaps with each other. Now, we want to show that they are also fully connected to each other. Again, let us proceed by contradiction. Consider Figure \ref{fig:indep-sets}-b. Suppose $w_1$ and $w_2$ are two different non-overlapping maximal independent sets. Suppose there exists an element in $w_2$ (call it $x_1^3$) which is connected to one of elements in $w_1$ (call it $x_1^1$) and is not connected to another element of $w_1$ (call it $x_1^2$). By using a similar discussion to the one in the previous paragraph, we may show that it is not possible. Thus, $x_1^3$ should be connected to $x_1^1$. Therefore, if for all $(x_1,x_2)\in\cX_1\times\cX_2$, $p(x_1,x_2)>0$, then maximal independent sets of  $\gxo$ are some separate fully connected sets. In other words, the complement of $\gxo$ is formed by some non-overlapping cliques. Finding the minimum entropy coloring of this graph is trivial and can be achieved by assigning different colors to these non-overlapping fully-connected maximal independent sets. 

This argument also holds for any power of $\gxo$. Suppose $\bx_1^1$, $\bx_1^2$ and $\bx_1^3$ are some typical sequences in $\cX_1^{n}$. If $\bx_1^1$ is not connected to $\bx_1^2$ and $\bx_1^3$, it is not possible to have $\bx_1^2$ and $\bx_1^3$ connected. Therefore, one can apply a similar argument to prove the theorem for $\gxon$, for some $n$. This completes the proof.  
\end{proof}

Here are some remarks about Theorem \ref{thm:color}:

\begin{itemize}
\item If the characteristic graph satisfying conditions of Theorem \ref{thm:color} is sparse, its power graph would also remain sparse (a sparse graph with $m$ vertices is a graph whose number of edges is much smaller than $\frac{m(m-1)}{2}$). 
\item The condition $p(x_1,x_2)>0$, for all $(x_1,x_2)\in\cX_1\times\cX_2$, is a necessary condition for Theorem \ref{thm:color}. In order to illustrate this, consider Figure \ref{fig:indep-set-example}. In this example, $x_1^1$, $x_1^2$ and $x_1^3$ are in $\cX_1$, and $x_2^1$, $x_2^2$ and $x_2^3$ are in $\cX_2$. Suppose $p(x_1^2,x_2^2)=0$. By considering the value of function $f$ at these points depicted in the figure, one can see that, in $\gxo$, $x_1^2$ is not connected to $x_1^1$ and $x_1^3$. However, $x_1^1$ and $x_1^3$ are connected to each other. Thus, Theorem \ref{thm:color} does not hold here. 
\item The condition used in Theorem \ref{thm:color} only restricts the probability distribution and it does not depend on the function $f$. Thus, for any function $f$ at the receiver, if we have a non-zero joint probability distribution of source random variables (for example, when source random variables are independent), finding the minimum-entropy coloring is easy and tractable. 
\end{itemize}

  \begin{figure}[t]
	\centering
    \includegraphics[width=8.5cm,height=5.5cm]{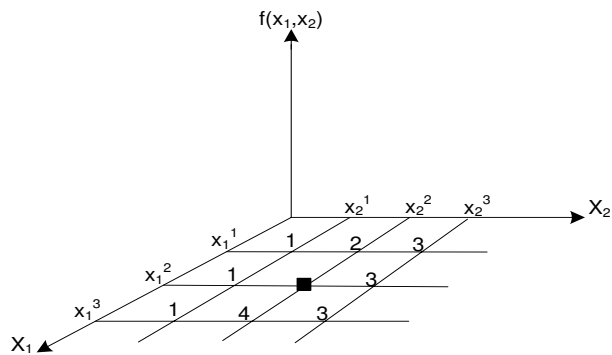}
    \caption{Having non-zero joint probability condition is necessary for Theorem \ref{thm:color}. A dark square represents a zero probability point.}
    \label{fig:indep-set-example}
  \end{figure}

\subsubsection{Quantization Functions}\label{subsec:special-functions}

In Section \ref{subsec:trivial-coloring}, we introduced a condition on the joint probability distribution of random variables which leads to a specific structure of the characteristic graph so that finding the minimum entropy coloring is not NP-hard. In this section, we consider some special functions which to lead to some special graph structures. 

An interesting function is a quantization function. A natural quantization function is a function which separates the $X_1-X_2$ plane into some rectangles such that each rectangle corresponds to a different value of that function. Sides of these rectangles are parallel to the plane axes. Figure \ref{fig:quantization-function}-a depicts such a quantization function.

Given a quantization function, one can extend different sides of each rectangle in the $X_1-X_2$ plane. This may make some new rectangles. We call each of them \textit{a function region}. Each function region can be determined by two subsets of $\cX_1$ and $\cX_2$. For example, in Figure \ref{fig:quantization-function}-b, one of the function regions is distinguished by the shaded area.   

\begin{defn}
Consider two function regions $\cX_1^1\times\cX_2^1$ and $\cX_1^2\times\cX_2^2$. If for any $x_1^1\in \cX_1^1$ and $x_1^2\in \cX_1^2$, there exist $x_2^1$ such that $p(x_1^1,x_2^1)p(x_1^2,x_2^1)>0$ and $f(x_1^1,x_2^1)\neq f(x_1^2,x_2^1)$, we say these two function regions are pairwise $X_1$-proper.
\end{defn}


\begin{thm}\label{thm:quan}
Consider a quantization function $f$ such that its function regions are pairwise $X_1$-proper. Then, $\gxo$ (and $\gxon$, for any $n$) is formed of some non-overlapping fully-connected maximal independent sets, and its minimum entropy coloring can be achieved by assigning different colors to different maximal independent sets.  
\end{thm}

  \begin{figure}[t]
	\centering
    \includegraphics[width=11cm,height=7cm]{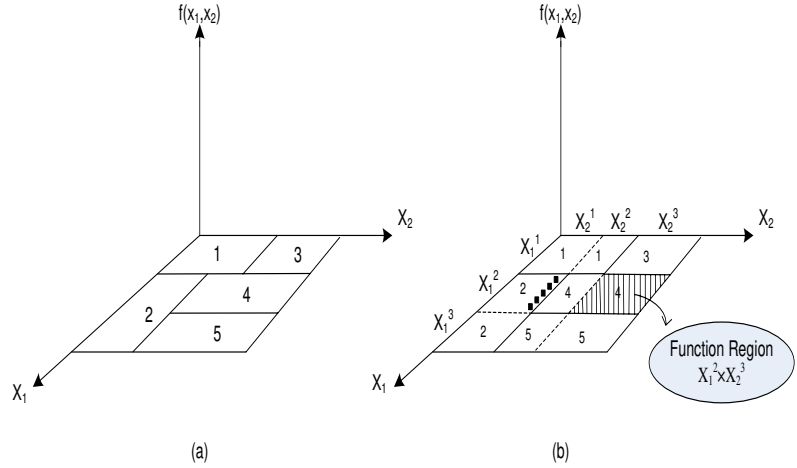}
    \caption{a) A quantization function. Function values are depicted in the figure on each rectangle. b) By extending sides of rectangles, the plane is covered by some function regions.}
    \label{fig:quantization-function}
  \end{figure}

\begin{proof}
We first prove it for $\gxo$. Suppose $\cX_1^1\times\cX_2^1$, and $\cX_1^2\times\cX_2^2$ are two $X_1$-proper function regions of a quantization function $f$, where $\cX_1^1\neq\cX_1^2$. We show that $\cX_1^1$ and $\cX_1^2$ are two non-overlapping fully-connected maximal independent sets. By definition, $\cX_1^1$ and $\cX_1^2$ are two non-equal partition sets of $\cX_1$. Thus, they do not have any element in common. 

Now, we want to show that vertices of each of these partition sets are not connected to each other. Without loss of generality, we show it for $\cX_1^1$. If this partition set of $\cX_1$ has only one element, this is a trivial case. So, suppose $x_1^1$ and $x_1^2$ are two elements in $\cX_1^1$. By definition of function regions, one can see that, for any $x_2^1\in\cX_2$ such that $p(x_1^1,x_2^1)p(x_1^2,x_2^1)>0$, then $f(x_1^1,x_2^1)=f(x_1^2,x_2^1)$. Thus, these two vertices are not connected to each other. Now, suppose $x_1^3$ is an element in $\cX_1^2$. Since these function regions are $X_1$-proper, there should exist at least one $x_2^1\in\cX_2$, such that $p(x_1^1,x_2^1)p(x_1^3,x_2^1)>0$, and $f(x_1^1,x_2^1)\neq f(x_1^3,x_2^1)$. Thus, $x_1^1$ and $x_1^3$ are connected to each other. Therefore, $\cX_1^1$ and $\cX_1^2$ are two non-overlapping fully-connected maximal independent sets. One can easily apply this argument to other partition sets. Thus, the minimum entropy coloring can be achieved by assigning different colors to different maximal independent sets (partition sets). The proof for $\gxon$, for any $n$, is similar to the one mentioned in Theorem \ref{thm:color}. This completes the proof.    
\end{proof}

Note that without $X_1$-proper condition of Theorem \ref{thm:quan}, assigning different colors to different partitions still leads to an achievable coloring scheme. However, it is not necessarily the minimum entropy coloring. In other words, without this condition, maximal independent sets may overlap.


\begin{cor}\label{thm:increasing}
If a function $f$ is strictly monotonic with respect to $X_1$, and $p(x_1,x_2)\neq 0$, for all $x_1\in\cX_1$ and $x_2\in\cX_2$, then, $\gxo$ (and, $\gxon$ for any $n$) is a complete graph. 
\end{cor}

Under conditions of Corollary \ref{thm:increasing}, functional compression does not give us any gain, because, in a complete graph, one should assign different colors to different vertices. Traditional compression where $f$ is the identity function is a special case of Corollary \ref{thm:increasing}.

\subsubsection{Minimum Entropy Coloring for an Arbitrary Graph}

Finding the minimum entropy coloring of an arbitrary graph (called the chromatic entropy) is NP-hard (\cite{np1}). Reference \cite{np1} showed that, even finding a coloring whose entropy is within $(\frac{1}{7}-\epsilon)\log m$ of its chromatic entropy is NP-hard, for any $\epsilon>0$, where $m$ is the number of vertices of the graph. That is a reason we introduced some special structures on the characteristic graph to have some tractable and practical schemes to find the minimum entropy coloring. While cases investigated in Sections \ref{subsec:trivial-coloring} and \ref{subsec:special-functions} cover certain practical cases, in this part, we want to consider this problem without assuming any special structure of the graph. In particular, we show that, by using a notion of an empirical Huffman coding scheme or a Lempel-Ziv coding scheme, one can heuristically relate finding the minimum-entropy coloring problem and finding the maximum independent set problem. While the minimum-entropy coloring problem is a recently studied problem, there are some heuristic algorithms to solve the maximum independent set problem \cite{algbook}.        

Suppose $\gxo$ is the characteristic graph of $X_1$. Without loss of generality, in this section, we consider $n=1$. All discussions can be extended to $\gxon$, for any $n$. Suppose $p(x_1)$ is the probability distribution of $X_1$. Let us define the adjacency matrix $A=[a_{ij}]$ for this graph as follows: $a_{ij}=1$ when $x_1^i$ and $x_1^j$ are connected to each other in $\gxo$, otherwise, $a_{ij}=0$. One can see that the adjacency matrix is symmetric, with all zeros in its diagonal. A \textit{one} in this matrix means that its corresponding vertices should be assigned to different colors.

Let us define a permutation matrix $P$ with the same size of $A$. This matrix has only a \textit{one} in each of its rows and columns. The matrix $PAP^t$ would be a matrix such that rows and columns of $A$ are reordered simultaneously, with respect to this permutation  matrix $P$. For any valid coloring, there exists a permutation matrix $P$, such that $PAP^t$ has zero square matrices on its diagonal. This reordering is such that, vertices with the same color are adjacent to each other in $PAP^t$. Each of these zero square matrices on the diagonal of $PAP^t$ represents a maximal independent set, or equivalently a color class. One can see that there exists a bijective mapping between any valid coloring and any permutation matrix $P$ which leads to have some zero square matrices on the diagonal of $PAP^t$.

\begin{ex}
For an example, consider the coloring of the graph depicted in Figure \ref{fig:graph-five-vertices}. This coloring leads to the following $PAP^t$ matrix:

\begin{equation}
PAP^t=\left( \begin{array}{ccc}
\begin{array}{ccc}
0 &0\\
0&0 \\
 \end{array} & D_1 &D_2\\
D_1&\begin{array}{cc}
0 &0\\
0&0 \\
 \end{array}&D_3\\
D_2&D_3& 0\\
\end{array}\right)
\end{equation}
where $D_i$, $i=1,2,3$ are non-zero matrices. Each of zero square matrices on the diagonal represents a color class, or a maximal independent set of this graph. The permutation matrix $P$ in this case is,
\begin{equation}
P=\left( \begin{array}{ccccc}
0&0&1&0&0\\
0&0&0&1&0\\
1&0&0&0&0\\
0&1&0&0&0\\
0&0&0&0&1\\
\end{array}\right).
\end{equation}
\end{ex}

Now, we want to take the probability distribution into account. To do this, we repeat each vertex $x_1^i$ in the adjacency matrix, $n_i$ times, such that $\frac{p(x_1^i)}{p(x_1^j)}=\frac{n_i}{n_j}$, for any valid $i$ and $j$. We call the achieved matrix, the weighted adjacency matrix and denote it by $A_w$. The above argument about the permutation remains the same. Any valid coloring can be represented by a permutation matrix $P$ such that $PA_wP^t$ has some zero square matrices on its diagonal. Since we represent the probability distribution of each vertex as its number of repetitions in $A_w$, the proportional sizes of zero square matrices on the diagonal of $PA_wP^t$ represent the corresponding probability of that color class. In other words, a color class of a larger zero square matrix has more probability than a color class with a smaller zero square matrix.    

Now, one can heuristically use Huffman coding technique to find a coloring (or its corresponding permutation matrix) to minimize the entropy.  To do this, we first find a permutation matrix which leads to the largest zero square matrix on the diagonal of $PA_wP^t$. Then, we assign a color to that class, and eliminate its corresponding rows and columns. We repeat this algorithm till all vertices are assigned to some colors. One can see that, finding the largest zero square matrix on the diagonal of $PA_wP^t$ is equivalent to finding the maximum independent set of a graph. Note that, it is a heuristic algorithm, and does not necessarily reach to the minimum entropy coloring. The other point is that, here, we have assumed that the probability distribution of $X_1$ is known. If we do not know this probability distribution, one can use an empirical distribution, instead of the actual distribution. In that case, using a Lempel-Ziv coding notion instead of Huffman coding leads to a similar algorithm.

\section{Conclusions and Future Work}\label{chap:conc}

In this paper, we considered different aspects of the functional compression problem where computing a function (or, some functions) of sources is desired at the receiver(s). The rate region of this problem has been considered in the literature under certain restrictive assumptions, particularly in terms of the network topology, the functions and the characteristics of the sources. In this paper, we significantly relaxed these assumptions. In Section \ref{chap:tree} of this paper, we considered this problem for an arbitrary tree network and asymptotically lossless computation and derived rate lower bounds. We showed that, for one stage tree networks with correlated sources, or for general trees with independent sources, these lower bounds are tight. For these cases, we proposed a modularized coding scheme based on graph colorings and Slepian-Wolf compression which performs arbitrarily closely to rate lower bounds. Optimal computations that should be performed at intermediate nodes are derived, for a general tree network with independent sources. We showed that, for a family of functions and random variables called chain rule proper sets, computation at intermediate nodes is not necessary. We also introduced a new condition on colorings of source random variables' characteristic graphs called the coloring connectivity condition (C.C.C.) and showed that, unlike the condition mentioned in Doshi \textit{et al.}, this condition is necessary and sufficient for any achievable coding scheme based on colorings. We also showed that, unlike entropy, graph entropy does not satisfy the chain rule.

The problem of having different desired functions with side information at the receiver was considered in Section \ref{chap:multi}. For this problem, we defined a new
concept named multi-functional graph entropy, an
extension of graph entropy defined by K\"orner to the multi-functional case. We showed that, the minimum achievable rate for this problem with side information is equal to conditional multi-functional graph entropy of the source random variable given the side information. We also proposed a coding scheme based on graph colorings to achieve this rate. 

In Section \ref{chap:feedback}, we investigated the effect of having feedback on the rate region of the functional compression problem. If the function at the receiver is the identity function, this problem reduces to the Slepian-Wolf compression with feedback, for which having feedback does not increase the rate. However, we showed that, in general, feedback can improve rate bounds.     

The problem of distributed functional compression with distortion was investigated in Section \ref{chap:distortion}. The objective is to compress correlated discrete sources so that an arbitrary deterministic function of those sources can be computed within a distortion level at the receiver. In this case, we computed a rate-distortion region for this problem which is not a single letter characterization. Then, we proposed a simple suboptimal coding scheme with a non-trivial performance guarantee. 

In these proposed coding schemes, one needs to compute the minimum entropy coloring of a characteristic graph. In general, finding this coloring is an NP-hard problem (\cite{np1}). However, in Section \ref{chap:min-coloring}, we showed that depending on the characteristic graph's structure, there are certain cases where finding the minimum entropy coloring is not NP-hard, but tractable and practical. In one of these cases, we showed that, by having a non-zero joint probability condition on random variables' distributions, for any desired function, finding the minimum entropy coloring can be solved in polynomial time. In another case, we showed that, if the desired function is a type of quantization functions, this problem is also tractable.

For possible future work, one may consider a general network topology rather than tree networks. For instance, one can consider a general multi-source multicast network in which receivers desire to have a deterministic function of source random variables. For the case of having the identity function at the receivers, this problem is well-studied in \cite{ahl2000}, \cite{medard2003} and \cite{random} under the name of network coding for multi-source multicast networks. Reference \cite{random} shows that, random linear network coding can perform arbitrarily closely to min-cut max-flow bounds. To have an achievable scheme for the functional version of this problem, one may perform random network coding on coloring random variables satisfying C.C.C. If receivers desire different functions, one can use colorings of multi-functional characteristic graphs satisfying C.C.C., and then use random network coding for these coloring random variables. This achievable scheme can be extended to disjoint multicast and disjoint multicast plus multicast cases described in \cite{medard2003}. This scheme is an achievable scheme; however it is not optimal in general. If sources are independent, one may use encoding/decoding functions derived for tree networks at intermediate nodes, along with network coding.     

Throughout this paper, we considered asymptotically lossless or lossy computation of a function. For possible future work, one may consider this problem for the zero-error computation of a function which leads to a communication complexity problem. One can use tools and schemes we have developed in this paper to attain some achievable schemes in the zero error computation case.

\end{document}